\documentclass{article}

\PassOptionsToPackage{numbers,sort&compress,square}{natbib}
 \usepackage[preprint]{neurips_2026}

\usepackage[utf8]{inputenc} 
\usepackage[T1]{fontenc}    
\usepackage{hyperref}       
\usepackage{url}            
\usepackage{booktabs}       
\usepackage{amsfonts}       
\usepackage{nicefrac}       
\usepackage{microtype}      
\usepackage{xcolor}         

\usepackage{graphicx} 
\usepackage{xcolor}
\usepackage{amsmath}
\usepackage{booktabs}
\usepackage{amsfonts}
\usepackage{algorithm}
\usepackage{algpseudocode}
\usepackage{float}  
\usepackage{placeins}

\usepackage{amsmath,amsfonts,amsthm,bm}
\usepackage{dsfont}
\usepackage{color}
\usepackage{mathtools}

\newcommand{\captiona}{{\em (a)}}
\newcommand{\captionb}{{\em (b)}}
\newcommand{\captionc}{{\em (c)}}
\newcommand{\captiond}{{\em (d)}}

\def\eqref#1{equation~\ref{#1}}

\def\ceil#1{\lceil #1 \rceil}

\def\1{\bm{1}}

\def\eps{{\epsilon}}

\def\mA{{\bm{A}}}

\DeclareMathAlphabet{\mathsfit}{\encodingdefault}{\sfdefault}{m}{sl}
\SetMathAlphabet{\mathsfit}{bold}{\encodingdefault}{\sfdefault}{bx}{n}

\newcommand{\E}{\mathbb{E}}

\newcommand{\R}{\mathbb{R}}

\DeclareMathOperator*{\argmax}{arg\,max}
\DeclareMathOperator*{\argmin}{arg\,min}

\let\hat\widehat

\def\given{{\,|\,}}

\newcommand{\bb}{\bm{b}}

\newcommand{\cB}{\mathcal{B}}

\newcommand{\cI}{\mathcal{I}}

\newcommand{\rbr}[1]{\left(#1\right)}

\newcommand{\cbr}[1]{\left\{#1\right\}}

\newcommand{\abr}[1]{\left|#1\right|}

\newtheorem{theorem}{Theorem}

\newtheorem{lemma}{Lemma}
\newtheorem{proposition}{Proposition}

\newtheorem{remark}{Remark}

\newtheorem{example}{Example}

\newtheorem{assumption}{Assumption}

\usepackage{hyperref}       
\usepackage{cleveref}       
\usepackage{crossreftools}
\Crefname{assumption}{Assumption}{Assumptions}   

\usepackage{amsmath,amssymb}

\usepackage{multirow}
\usepackage{graphicx}
\usepackage{spverbatim}

\algtext*{EndIf}
\algtext*{EndWhile}
\algtext*{EndFor}
\algtext*{EndLoop}
\algtext*{EndProcedure}
\algtext*{EndFunction}

\newcommand{\hq}{\hat{q}}
\newcommand{\dprior}{\delta_{\mu_0}}
\newcommand{\upbias}{\overline{\alpha}}
\newcommand{\underbias}{\underline{\alpha}}
\newcommand{\up}{P_{\mathrm{info}}}
\newcommand{\vertex}{v}
\newcommand{\movv}{\mathrm{MOV}}
\newcommand{\dist}{\operatorname{dist}}
\newcommand{\rank}{\operatorname{rank}}
\newcommand{\ones}{\mathbf{1}}
\newcommand{\pos}[1]{\left(#1\right)_{+}}

\newcommand{\states}{\Omega}
\newcommand{\actions}{A}
\newcommand{\signals}{S}
\newcommand{\simplex}[1]{\Delta(#1)}
\newcommand{\prior}{\mu_0}

\newcommand{\bias}{\alpha}
\newcommand{\tbias}{\alpha^{\star}}
\newcommand{\mech}{\pi}          
\newcommand{\postdist}{\rho}     
\newcommand{\scheme}{\tau}       
\newcommand{\post}{\nu}
\newcommand{\dpost}{\hat{\nu}}
\newcommand{\bestresp}[2]{a^\star\!\left(#1;#2\right)}
\newcommand{\areg}[2]{R_{#1}^{#2}}

\newcommand{\defaulta}{a_0}
\newcommand{\Prob}{\mathbb{P}}

\newcommand{\OPT}{\operatorname{OPT}}
\newcommand{\Reg}{\operatorname{Reg}}
\newcommand{\qth}{q_{\mathrm{th}}} 

\newcommand{\ivalue}[3]{V_{#1}^{#2}\!\left(#3\right)}
\newcommand{\iopt}[2]{\OPT_{#1}\!\left(#2\right)}
\newcommand{\ireg}[4]{\Reg_{#2}^{#1}\!\left(#3;#4\right)}

\newcommand{\Cbs}{C_{\mathrm{BS}}}
\newcommand{\Cse}{C_{\mathrm{SE}}}
\newcommand{\Clb}{c_{\mathrm{LB}}}
\newcommand{\Cgse}{C_{\mathrm{GSE}}}
\newcommand{\Cloc}{C_{\mathrm{loc}}}
\newcommand{\Cphase}{C_{\mathrm{phase}}}
\newcommand{\Crep}{C_{\mathrm{rep}}}
\newcommand{\Kdist}{K_{\mathrm{dist}}}
\newcommand{\Ksafe}{K_{\mathrm{safe}}}
\newcommand{\Kprobe}{K_{\mathrm{probe}}}
\newcommand{\Lnu}{L_{\nu}}
\newcommand{\pmin}{p_{\min}}

\newcommand{\Gmax}{G_{\max}}
\newcommand{\Dumax}{\Delta U_{\max}}
\newcommand{\cinf}{c_{\mathrm{info}}}

\title{Learning to Persuade a Biased Receiver}

\author{
  Yuqi Pan \\
  Harvard University \\
  \texttt{yuqipan@g.harvard.edu}
  \And
  Sadie Zhao \\
  Harvard University \\
  \texttt{sadie\_zhao@g.harvard.edu}
  \AND
  Milind Tambe \\
  Harvard University \\
  \texttt{milind\_tambe@harvard.edu}
  \And
  Yiling Chen \\
  Harvard University \\
  \texttt{yiling@seas.harvard.edu}
}

\begin{document}

\maketitle

\begin{abstract}
  We study a repeated information design setting in which the receiver, who is also the decision-maker, updates beliefs in a systematically biased way. More specifically, a distorted posterior in our model can be written as a convex combination of the prior and the Bayesian posterior, governed by a fixed but unknown parameter. Over repeated interactions, the sender chooses persuasive signaling schemes, observes only the receiver’s realized actions, and seeks to minimize regret relative to a full-information oracle that knows the receiver’s biased updating rule.
 We propose a safe exploration algorithm for learning the receiver’s bias while maintaining high persuasion value. The algorithm exploits the asymmetric cost of probing: conservative probes incur only local loss, whereas overly aggressive probes may lose the persuasive opportunity entirely.  For general finite state and action spaces and arbitrary bounded utilities, our method achieves $O(\log\log T)$ regret. A matching $\Omega(\log\log T)$ lower bound shows that this rate is optimal. We further discuss the influence on receiver welfare, as well as extensions to jointly unknown prior and bias, and contextual settings with time-varying priors and utilities.
\end{abstract}
\section{Introduction}
\label{sec:intro}
Information designers, such as AI platforms, LLM-based assistants, and automated warning systems, often observe utility-relevant state information and decide how to present it to a human decision-maker (receiver). 
Effective communication depends not only on how informative the signal is, but also on how the receiver updates her belief in response to algorithmic evidence. 
Classical Bayesian persuasion models the receiver as updating by Bayes' rule \cite{kamenica2011bayesian}. Yet when algorithmic systems advise human receivers, they may underweight algorithmic advice, consistent with algorithm aversion phenomenon \cite{dietvorst2015algorithm,mahmud2022influences}: humans may anchor on prior judgments or treat algorithmic explanations as only partially credible. Thus the same message can induce different actions across receivers, not because the evidence differs, but because receivers differ in how strongly they incorporate it \cite{dasgupta2020theory,bailey2023meta,ba2025over}.


As a running example, consider an online marketplace that communicates fraud-risk information to merchants through risk scores or warning labels. After seeing the message, a merchant decides whether to ship the order or request additional verification. The platform and the merchant both care about fraud, but their objectives are not identical: the platform also values marketplace-wide integrity and buyer trust, while each merchant trades off verification costs against fraud risk. Moreover, merchants may react differently to the same warning: one may treat it as strong evidence of fraud, while another may discount the warning and stay close to her prior assessment. If the platform knew a merchant's updating strength, it could calibrate messages accordingly.


Motivated by such settings, we study a repeated information-design problem with an unknown biased receiver.
In each of \(T\) rounds, the sender picks a signaling scheme that maps the state i.i.d. sampled from the prior to a randomized signal. 
After observing the realized signal, the receiver forms a biased posterior belief and takes a best-response action accordingly.
Here the receiver's biased updating rule is fixed across rounds but unknown to the sender, capturing repeated interactions with the same receiver or with a stable receiver type.
The sender observes this action and uses the history of past signals and actions to update future signaling schemes, whose goal is to achieve high long-run utility.


To quantify biased belief updating, we adopt a standard linear model \cite{epstein2010non,hagmann2017persuasion,tang2021bayesian, de2022non,chen2024bias}. If a signal induces a Bayesian posterior $\post$ from prior $\prior$, a receiver with bias level $\tbias\in[0,1]$ instead acts on a distorted belief $\tbias\post+(1-\tbias)\prior$.
The case $\tbias=1$ corresponds to Bayesian updating, while $\tbias=0$ means that the receiver completely ignores the signal and keeps the prior. 
Our setting is related to, but conceptually distinct from, existing work on persuasion with biased receivers \cite{de2022non,feng2024rationality,chen2024bias,kobayashi2025dynamic}.
Most closely, \citet{chen2024bias} study the same linear belief distortion, but their objective is diagnostic, i.e., deciding whether the receiver's bias exceeds a threshold. 
While in our setting, there exists a new exploration--exploitation tradeoff: a diagnostic signal may reveal \(\tbias\) but perform poorly for current utility, while a utility-maximizing signal may reveal little about how the receiver processes evidence. 


Moreover, feedback is indirect and discretized: an observed action is informative only when different candidate bias levels would induce different best responses under the realized posterior. The central challenge is therefore to learn the unknown updating strength from coarse action feedback without sacrificing too much persuasion utility along the way.

\paragraph{Our results.}
We first isolate the core structure in a binary-state, binary-action setting. There, the full-information optimal scheme is pinned down by the posterior cutoff at which the biased receiver switches actions. A standard binary search over this cutoff yields \(O(\log T)\) regret. Our main insight to obtain lower regret is that probing losses are asymmetric: a conservative probe that stays on the safe side still induces the desired action and loses only a local utility gap, whereas an aggressive probe may lose the persuasion opportunity entirely. 
This mirrors online posted-price auctions \citep{kleinberg2003value}. 
Exploiting this asymmetry, our Safe Exploration algorithm approaches the unknown cutoff using increasingly fine probes from the safe side, achieves \(O(\log\log T)\) regret.
We further prove a matching $\Omega(\log\log T)$ lower bound, which also implies an $\Omega(\log\log T)$ lower bound for the general setting. 
We then extend the safe-exploration idea to general finite state and action spaces with arbitrary bounded utilities. In this setting, there is no single posterior cutoff: each action corresponds to a polyhedral region in the Bayesian-posterior simplex, and these regions move with \(\tbias\). 
We construct interval-safe schemes whose support remains incentive-compatible throughout the current bias uncertainty interval, and perturb only binding movable constraints to create local one-dimensional probes while preserving Bayes plausibility. The resulting General Safe Exploration algorithm runs in polynomial time and achieves \(O(\log\log T)\) regret.
Finally, we discuss several extensions concerning receiver welfare, jointly unknown prior and bias, and contextual environments with time-varying priors and utilities.

\paragraph{Related work.}
Our work is closely related to the literature on non-Bayesian belief updating
\cite{edwards1968conservatism,bar1980base,grether1980bayes,benjamin2019errors,
ba2025over,taber2006motivated,cooney2019learning,cranford2020toward,tappin2020bayesian}, Bayesian persuasion and information design
\cite{kamenica2011bayesian,rayo2010optimal,gentzkow2016rothschild,bergemann2016information,dughmi2016algorithmic,dughmi2017algorithmic}, persuasion with biased receivers
\cite{de2022non,tang2021bayesian,feng2024rationality,chen2024bias,kobayashi2025dynamic},
and learning in information design
\cite{castiglioni2020online,castiglioni2021multi,feng2022online,harris2023algorithmic,
zu2021learning,Li_Lin_2025,wu2022sequential,bacchiocchi2024markov}.  Due to space limits, we defer a detailed discussion to Appendix~\ref{app:related-work}.

\section{Problem Setup}
\label{sec:model}
\paragraph{Primitives.}
Let $\states$ and $\actions$ be finite state and action spaces, and
$\signals$ a signal space. The common prior is
$\prior \in \simplex{\states}$. 
Sender and receiver utility functions
\(u_S,u_R:\actions\times\states\to\R\) are publicly known.
We write $\cI =(\states,\actions,\prior,u_S,u_R)$
for the known problem instance. The receiver's true bias
\(\tbias\in(0,1]\) is fixed but unknown to the sender.

A signaling scheme is a map \(\mech:\states\to\simplex{\signals}\).
The realized state \(\omega\in\states\) is observed by the sender but not the receiver.
The sender then draws a signal \(s\in\signals\) following \(\mech(\cdot\mid \omega)\) and sends it to the receiver. 
A standard Bayesian receiver forms the posterior belief \(\post_s\in\simplex{\states}\) from the signal $s$, where
\(
 \post_s(\omega)=
 \Prob[\omega| s]
 =
 \nicefrac{\prior(\omega)\mech(s\mid \omega)}
 {\sum_{\omega'\in\states}\prior(\omega')\mech(s\mid \omega')} .
\)
Equivalently, a signaling scheme can be viewed as a distribution over posteriors $\postdist \in \simplex{\simplex{\states}}$. 
We now introduce the well-known splitting lemma.
\begin{lemma}[Splitting Lemma \cite{kamenica2011bayesian}]
\label{lemma:splitting}
A distribution of posteriors $\postdist$ is Bayes-plausible (i.e., can be induced by some signaling scheme $\mech$) iff the expected posterior equals the prior:
$
\E_{\post \sim \postdist}[\post] = \prior.
$
\end{lemma}
Therefore, the sender equivalently chooses a Bayes-plausible distribution \(\postdist\in\simplex{\simplex{\states}}\) satisfying
\(\E_{\post\sim\postdist}[\post]=\prior\). 
We adopt this posterior-based view throughout the paper.

\paragraph{Biased belief updating.}
We consider an $\tbias$-biased receiver, where bias level \(\tbias\in(0,1]\) is fixed across rounds but unknown to the sender.
We adopt the widely-used linear distortion model \cite{epstein2010non,hagmann2017persuasion,tang2021bayesian, de2022non,chen2024bias}, which captures partial incorporation of signal evidence. Mathematically, given Bayesian posterior $\post$ and bias level $\bias \in [0,1]$, the receiver produces a distorted posterior $\dpost=(1-\bias)\prior + \bias \post$, and chooses the corresponding best action:
$$
\bestresp{\post}{\bias}
\in
\argmax_{a \in \actions}
\sum_{\omega \in \states} \left((1-\bias)\prior(\omega) + \bias \post(\omega)\right)u_R(a,\omega),
$$
with a fixed tie-breaking rule satisfying mild conditions, specified in \Cref{sec:general}.
Throughout the paper, we reserve $\post$ for Bayesian posteriors and $\dpost$ for distorted posteriors. Note that although the receiver optimizes with respect to the distorted posterior $\dpost$, the sender can only design the Bayesian posterior $\post$, so we always discuss the action region within the Bayesian posterior space.

\paragraph{Timing.}
The interaction proceeds for $T$ rounds.
Let \(h_0=\emptyset\) and $h_t = \{(\postdist_\tau,\omega_\tau,\post_\tau,a_\tau)\}_{\tau=1}^t$
denote the sender's history up to the end of round $t$, consisting of previous signaling schemes, realized states, induced posteriors, and receiver's adopted actions.
Let $\Pi$ denote the sender's adaptive policy, which is a mapping from histories to Bayes-plausible distributions over Bayesian posteriors. $\Pi$ may depend on $\cI$ and $T$, but not the unknown $\tbias$.
In $t = 1,2,\dots,T$, the interaction proceeds as follows:
\begin{enumerate}
    \item 
    Based on \(h_{t-1}\), the sender publicly commits to a Bayes-plausible distribution (equivalently, a signaling scheme) over posteriors
    $\postdist_t = \Pi(h_{t-1})$.
    \item 
    A state is drawn according to the prior:
    $
    \omega_t \sim \prior.
    $
    The sender observes $\omega_t$.

    \item 
    Given $\omega_t$, the sender draws a signal from the chosen signaling scheme and sends it to the receiver.
    The realized signal induces a posterior belief $\post_t \in \simplex{\states}$.

    \item 
    After observing the signal, the receiver forms the distorted posterior
    $(1-\tbias)\prior + \tbias \post_t
    $
    and chooses an action
    $
    a_t = \bestresp{\post_t}{\tbias}.
    $

    \item 
    The sender and receiver obtain utility
    $
    u_S(a_t,\omega_t)
    $
    and
    $
    u_R(a_t,\omega_t).
    $
    The sender observes the receiver's action $a_t$, and the history is updated to $h_t$.
\end{enumerate}

\paragraph{Objective and regret.}
For a fixed instance $\cI$ and bias $\bias$, define the sender's
reduced-form expected utility at Bayesian posterior $\post$ as
\[
\ivalue{\cI}{\bias}{\post}
=
\sum_{\omega\in\states}
\post(\omega)
u_S\!\left(\bestresp{\post}{\bias},\omega\right).
\]
The full-information one-round benchmark knows \(\tbias\) and plays an optimal Bayes-plausible posterior distribution in every round, which is defined as:
\[
\iopt{\cI}{\tbias}
=
\sup_{\postdist\in\simplex{\simplex{\states}}:
\E_{\post\sim\postdist}[\post]=\prior}
\E_{\post\sim\postdist}
\left[
\ivalue{\cI}{\tbias}{\post}
\right]
\]
Let \(\Prob_{t,\cI,\tbias}^{\Pi}\) denote the distribution of history 
\(h_t\) under \(\Pi,\cI,\tbias\). The expected regret is
\[
 \ireg{\cI}{T}{\Pi}{\tbias}
 =
 T\iopt{\cI}{\tbias}
 -
 \sum_{t=1}^T
 \E_{h_{t-1}\sim \Prob_{t-1,\cI,\tbias}^{\Pi}}
 \left[
 \E_{\post_t\sim\Pi(h_{t-1})}
 \left[
 \ivalue{\cI}{\tbias}{\post_t}
 \right]
 \right].
\]

\section{Warm-up: Safe Exploration via a One-Dimensional Threshold}\label{sec:binary}
\begin{figure}[h]
    \centering

    \begin{minipage}[t]{0.485\linewidth}
        \centering
        \includegraphics[width=\linewidth]{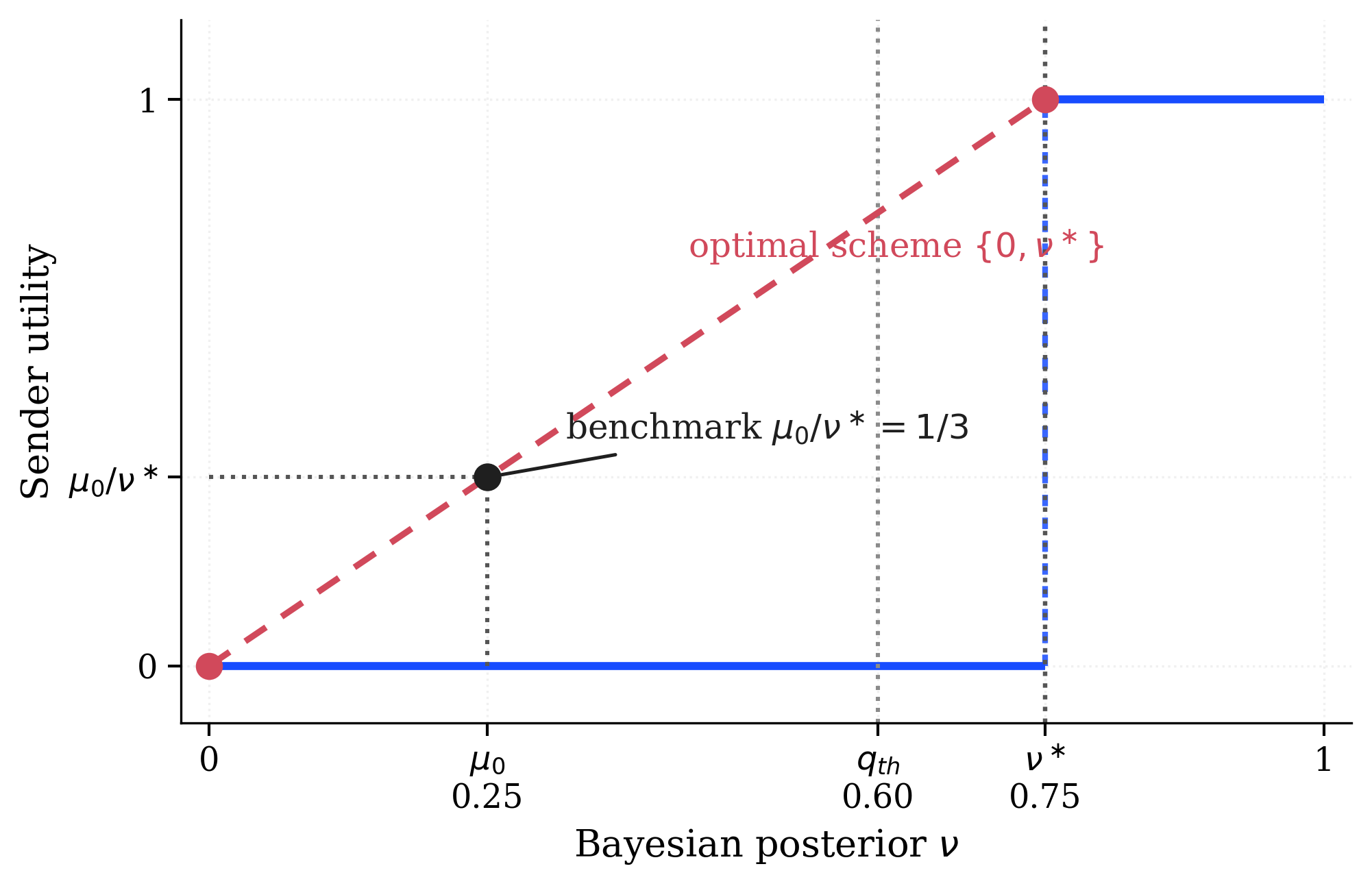}
        \\[-0.4ex]
        {\footnotesize \captiona\ Full-information benchmark}
    \end{minipage}
    \hfill
    \begin{minipage}[t]{0.485\linewidth}
        \centering
        \includegraphics[width=\linewidth]{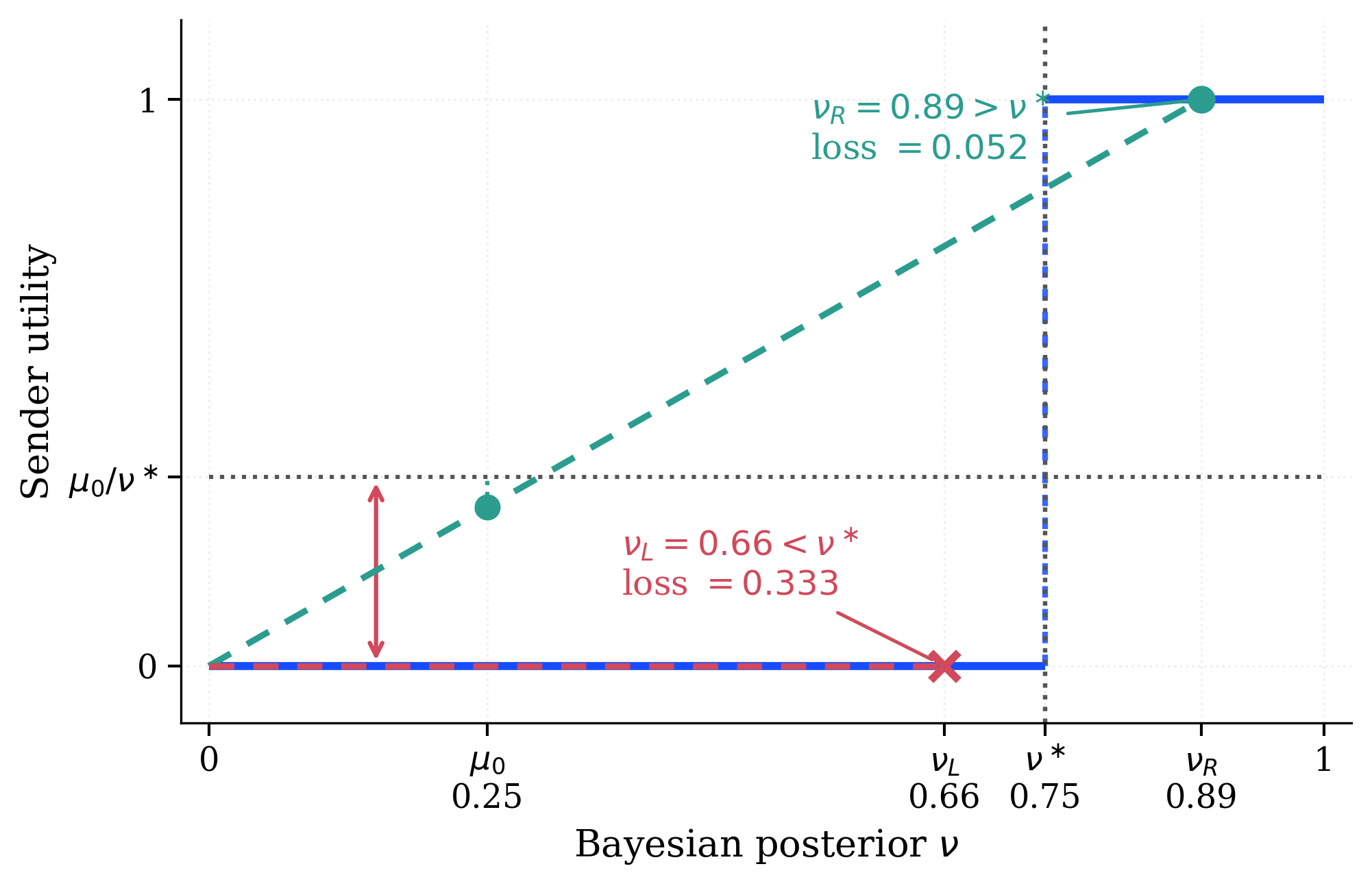}
        \\[-0.4ex]
        {\footnotesize \captionb\ Asymmetric probe loss}
    \end{minipage}

    \vspace{0.8ex}

    \begin{minipage}[t]{0.485\linewidth}
        \centering
        \includegraphics[width=\linewidth]
         {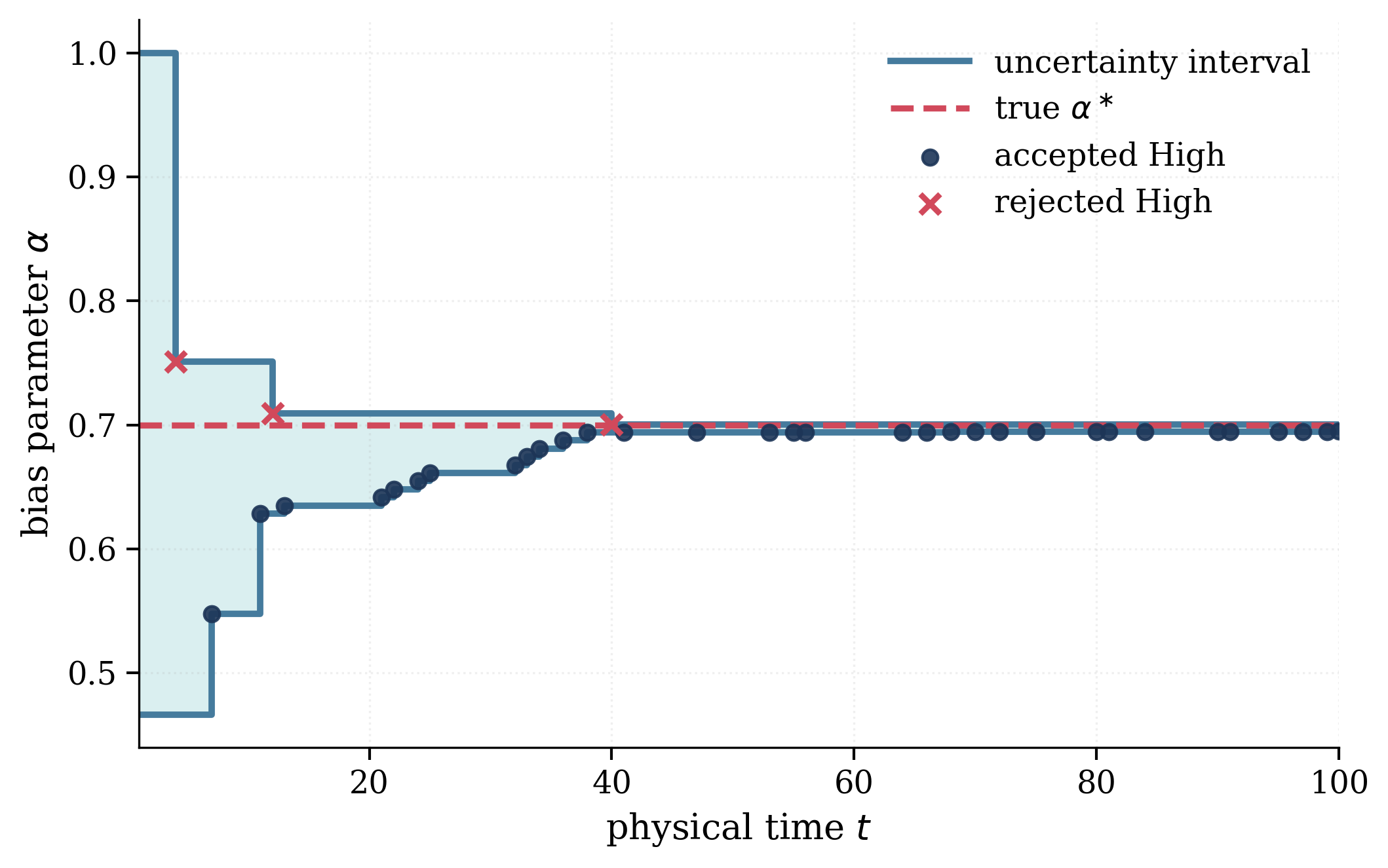}
        \\[-0.4ex]
        {\footnotesize \captionc\  Safe-side probing sample path}
    \end{minipage}
    \hfill
    \begin{minipage}[t]{0.485\linewidth}
        \centering
        \includegraphics[width=\linewidth]
        {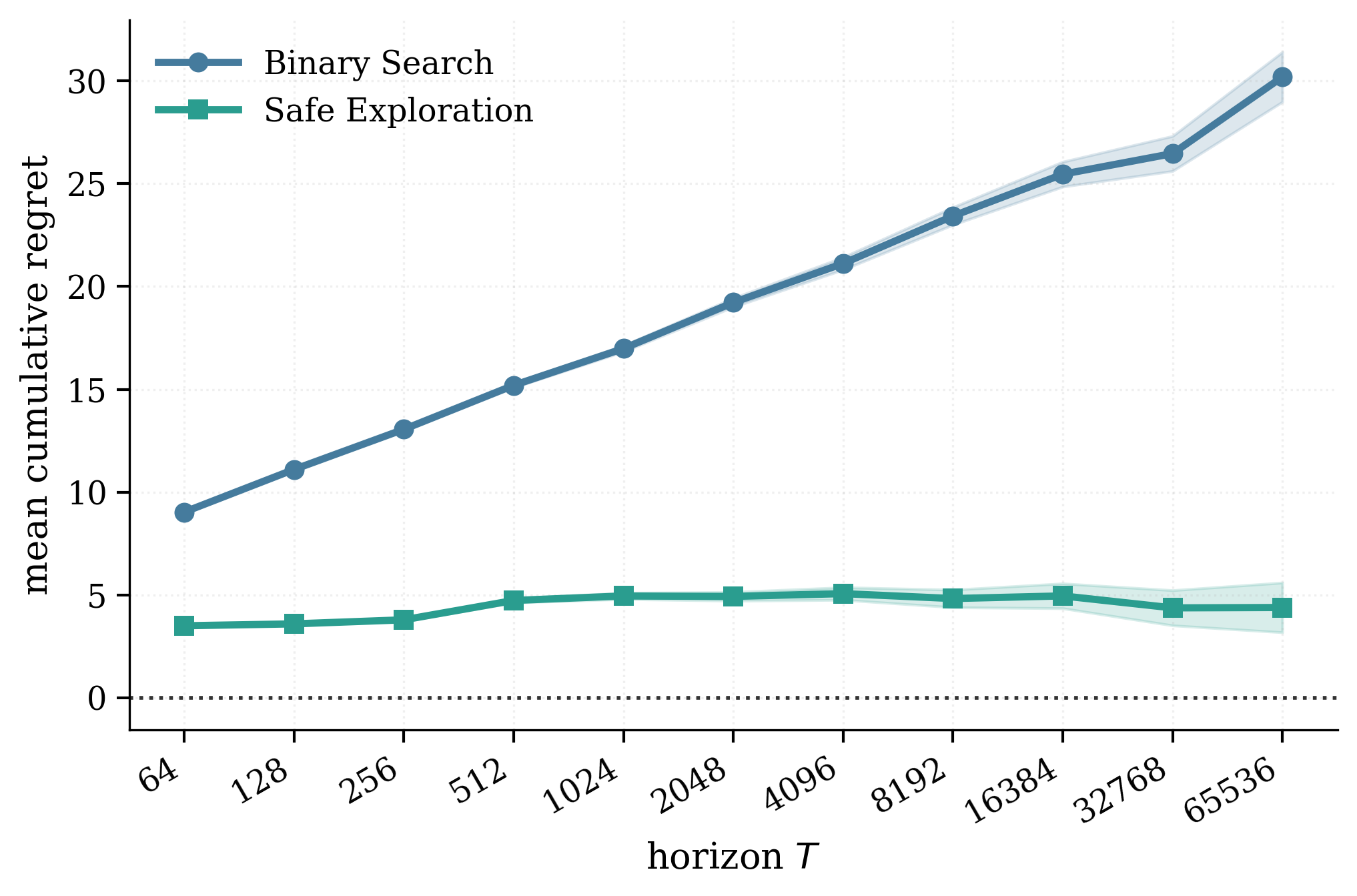}
        \\[-0.4ex]
        {\footnotesize \captiond\ $O(\log T)$ regret vs.\ $O(\log\log T)$ regret}
    \end{minipage}

    \caption{
    Binary instance
    $B=(\prior,\qth)=(0.25,0.60)$ and true bias $\tbias=0.70$.
    }
    \label{fig:binary-overview}
    \vspace{-1ex}
\end{figure}
We begin with a binary specialization that isolates the the main learning structure. Let $\states=\actions=\{0,1\}$. With a slight abuse of notation, write $\prior\in(0,1)$ for $\mu_0(1)$ and $\post\in[0,1]$ for $\post(1)$.
Moreover, we consider the canonical sender utility
\(u_S(a,\omega)=\mathds{1}\{a=1\}\), so the sender only cares about inducing action $1$.
On the receiver side, action $0$ is set as optimal at the prior without loss of generality, otherwise the sender can induce action $1$ without learning and regret is zero.

Since the receiver's utility difference between actions $1$ and $0$ is affine in the distorted belief, the action follows a cutoff rule. 
Let $\qth\in(\prior,1)$ denote the corresponding cutoff: action $1$ is chosen iff
\(
(1-\bias)\prior+\bias\post \ge \qth .
\)
Equivalently, this induces the Bayesian-posterior cutoff:
$
\post_{\cB}(\bias)
=
\prior+(\qth-\prior)/\bias$. Define $\areg{1}{\bias}=[\post_{\cB}(\bias),1]$ and $\areg{0}{\bias}=[0,\post_{\cB}(\bias)]$ as the  Bayesian posterior regions that weakly induce action $1$ and $0$ respectively.
For true bias level $\tbias$, we write $\post^\star = \post_{\cB}(\tbias)$.
\Cref{fig:binary-overview}(a) illustrates this structure, where the blue step function is the sender's reduced-form utility function over Bayesian posterior. We denote the binary instance by \(\cB=(\prior,\qth)\).

If $\tbias<\alpha_{\min}(\cB):=(\qth-\prior)/(1-\prior)$, then
$\areg{1}{\tbias}=\emptyset$ and persuasion is infeasible, so the regret is trivially zero. 
We focus on $\tbias\in[\alpha_{\min}(\cB),1]$, for any $\bias$, define corresponding two-point signaling scheme $\scheme(\bias)$
as the Bayes-plausible distribution over the two posteriors $\{0, \post_{\cB}(\bias)\}$ that places mass $1 - \prior / \post_{\cB}(\bias)$ on $0$ and $\prior / \post_{\cB}(\bias)$ on $\post_{\cB}(\bias)$.
The full-information optimum is attained by $\scheme(\tbias)$, with value \(\prior/\post_{\cB}(\tbias)\).
Intuitively, any posterior strictly above \(\post^\star\) is wasteful because it induces the same action while using more positive posterior mass in the Bayes-plausibility constraint \(\mathbb{E}[\post]=\prior\). Any posterior below \(\post^\star\) is also wasteful because it produces no utility while still using posterior mass. Thus the optimal scheme puts persuasive mass at the lowest useful posterior \(\post^\star\) and offsets it with the lowest posterior \(0\), which is the red line in \Cref{fig:binary-overview}(a). Learning the optimal scheme is therefore reduced to learning the unknown threshold \(\post^\star\).



%

\paragraph{Binary search (BS).}
Since \(\post_{\cB}(\alpha)\) is monotone in \(\alpha\), we equivalently maintain
an uncertainty interval for the unknown bias \(\alpha^\star\).
Let $J=[\underbias,\upbias]\in\tbias$ denote the current uncertain bias interval. 
Given \(J\), the natural safe choice is obtained by
restricting attention to posteriors that induce action \(1\) for every \(\alpha\in J\).
Define \emph{interval-safe region} $\areg{1}{J}=\bigcap_{\bias\in J}\areg{1}{\bias}$, i.e., the region of Bayesian posteriors that can robustly induce action $1$ within $J$.
Since \(\post_\cB(\bias)\) is decreasing in \(\bias\), $\areg{1}{J}=\areg{1}{\underbias}=[\post_\cB(\underbias),1]$.
Hence scheme \(\scheme(\underbias)\) is a safe choice: its high posterior robustly induces ideal action \(1\) for every \(\bias\in J\), with regret $\prior/\post_{\cB}(\tbias)-\prior/\post_{\cB}(\underbias)$, which vanishes as \(\underbias \uparrow \tbias\).

To achieve low regret, we shrink \(J\) by linking interval updates to realized actions. If the sender plays \(\scheme(m)\), then upon realizing the high posterior \(\post_{\cB}(m)\), the receiver chooses action \(1\) iff
$
\post_{\cB}(m)\ge \post_{\cB}(\tbias)
\Leftrightarrow
m\le \tbias,
$
Thus the sender can update \(J\) to either \([m,\upbias]\) or \([\underbias,m]\).
A natural algorithm is binary search: probe \(m=(\underbias+\upbias)/2\) by playing \(\scheme(m)\) until \(\post_{\cB}(m)\) is realized. Each informative realization halves \(J\). After \(O(\log T)\) such realizations, \(|J|=O(1/T)\), so the sender commits to the safe scheme \(\scheme(\underbias)\) thereafter. The full algorithm is in \Cref{app:binary-baseline} and guarantees \(O(\log T)\) regret.


\begin{proposition}\label{prop:binary-bsearch}
For every binary instance \(\cB=(\prior,\qth)\) and every bias level \(\tbias\in(0,1]\), $\ireg{T}{\cB}{\Pi^{\mathrm{BS}}}{\tbias}=O(\log T)$ holds with a hidden constant \(\Cbs(\cB)\).
\end{proposition}

Throughout the paper, all asymptotic bounds are pointwise in the fixed learning instance. More precisely, the constant is fixed once \(I=(\states,\actions,\prior,u_S,u_R)\) and \(\tbias\) are fixed, and it never depends on the horizon \(T\). The constants in \Cref{sec:binary,sec:general} are collected in \Cref{app:binary-constants,app:general-constants} respectively.

\paragraph{Safe exploration (SE).}
Binary search is suboptimal because it ignores the asymmetric losses on the two sides of a probe, see \Cref{fig:binary-overview}(b). 
For a candidate $m\in J$ with probing scheme $\scheme(m)$, if $m\le \tbias$, the regret is only the local gap
$\prior/\post_{\cB}(\tbias)-\prior/\post_{\cB}(m)$, which vanishes as $m\uparrow\tbias$. 
If $m>\tbias$, however, the high posterior $\post_{\cB}(m)<\post^*$ no longer induces action $1$, causing regret $\prior/\post_{\cB}(\tbias)$.
Thus, a conservative probe loses only a local utility gap, whereas an overly aggressive probe may lose the entire persuasive opportunity.
This mirrors online posted-price auction \cite{kleinberg2003value}, where a seller repeatedly posts a price to a buyer with unknown valuation. There, posting below the buyer's value only loses the price gap, whereas posting above it loses the sale entirely. 

Our \emph{Safe Exploration} algorithm, \Cref{alg:astc}, exploits this asymmetry by probing from the safe side with a small step. In each phase, starting from \(J=[\underbias,\upbias]\), it scans
\(m=\underbias+\epsilon,\underbias+2\epsilon,\ldots\)
from left to right and plays \(\scheme(m)\), where \(\epsilon\) is initialized to \(1/2\) and squared after each interval update. The scan stops at the first informative realization for which the high posterior induces action \(0\), and updates \(J\) to the adjacent bracket \([m_{\mathrm{prev}},m]\). If no switch occurs before the next probe exceeds \(\upbias\), it updates \(J\) to \([m_{\mathrm{prev}},\upbias]\).
\Cref{fig:binary-overview}(c) illustrates the safe-side nature of the algorithm: each phase scans upward from the conservative endpoint, accumulating many low-loss safe probes near $\tbias$ before at most one aggressive probe closes the phase. 
Since the new interval length is at most the current step size and the next step size is squared, \(J\) shrinks doubly exponentially, so \(O(\log\log T)\) phases suffice to reach length \(1/T\). \Cref{alg:astc} guarantees \(O(\log\log T)\) regret in total. The simulation in \Cref{fig:binary-overview}(d) also matches the \(O(\log T)\) versus \(O(\log\log T)\) comparison.

Geometrically, each probe enlarges the interval-safe region \(\areg{1}{J}\). Starting from the safe scheme \(\scheme(\underbias)\), whose high posterior lies on the boundary, replacing \(\post(\underbias)\) with \(\post(m)\) pushes this boundary outward. As \(J\) shrinks, the interval-safe region expands toward the true action region \(\areg{1}{\tbias}\).

\begin{proposition}\label{prop:binary-regret-new}
For every binary instance \(\cB=(\prior,\qth)\) and every bias level \(\tbias\in(0,1]\), $\ireg{T}{\cB}{\Pi^{\mathrm{SE}}}{\tbias}=O(\log\log T)$ holds with a hidden constant \(\Cse(\cB)\).

\begin{algorithm}[h]
\caption{Safe Exploration (SE)}
\label{alg:astc}
\begin{algorithmic}[1]
\Require Horizon \(T\), binary instance $\cB=(\prior,\qth)$
\State $\underbias \gets \bias_{\min}(\cB)$, $\upbias\gets 1$, $\epsilon\gets 1/2$, $t\gets 1$

\Statex \textbf{Stage 1: Safe Exploration}
\While{\(\upbias-\underbias > T^{-1}\) and \(t\le T\)}
    \State  \(\; m_{\mathrm{prev}} \gets \underbias\), \(m \gets \underbias+\eps\)
    \While{\(m\le \upbias\) and \(t\le T\)}
        \State commit to and play \(\scheme(m)\)
        \If{the realized posterior is not \(0\)}
            \If{the receiver chooses action \(1\)}
                \State \(m_{\mathrm{prev}} \gets m\), \(\; m \gets m+\eps\)
            \Else
                \State \((\underbias,\upbias) \gets (m_{\mathrm{prev}},m)\), \(\;\eps \gets \eps^2\), \(\;t \gets t+1\)
                \State \textbf{break}
            \EndIf
        \EndIf
         \State \(t \gets t+1\)
    \EndWhile
    \If{\(m > \upbias\)}
        \State \((\underbias,\upbias) \gets (m_{\mathrm{prev}},\upbias)\), \(\;\eps \gets \eps^2\)
    \EndIf
\EndWhile

\Statex \textbf{Stage 2: Commitment}
\State Commit to \(\scheme(\underbias)\) for all remaining rounds.
\end{algorithmic}
\end{algorithm}

\end{proposition}



\paragraph{Lower bound.}
We now prove a matching \(\Omega(\log\log T)\) lower bound, showing that \Cref{alg:astc} is asymptotically optimal. The proof is inspired by \citet{kleinberg2003value}, but must handle a new difficulty: information about the unknown threshold arrives only through informative signal realizations, while regret is incurred in every physical round. We apply Yao's principle \cite{yao1977probabilistic} to reduce the minimax lower bound to deterministic sender strategies under a carefully chosen bias distribution, and then use martingale tools to translate the bound back to horizon \(T\). This yields the following theorem, which also gives an \(\Omega(\log\log T)\) lower bound for the general case in \Cref{sec:general}.

\begin{theorem}[lower bound]
\label{thm:lowerbound}
For any binary instance $\cB=(\prior,\qth)$ and any
randomized sender strategy $\Pi$, $\sup_{\tbias\in[0,1]}
\ireg{\cB}{T}{\Pi}{\tbias}=\Omega(\log\log T)$ holds with a hidden constant $\Clb(\cB)$.
\end{theorem}




\section{General Setting: Safe Exploration in Moving Posterior Geometry}\label{sec:general}
The binary case shows why safe exploration beats binary search: probing losses are asymmetric. 
We now similarly define interval-safe action region in the general setting.
For any action $a\in\actions$ and bias $\bias>0$, let $\areg{a}{\bias}$ be the region of Bayesian posteriors that weakly induce action $a$:
\[
\areg{a}{\bias}
=
\left\{
\post\in\simplex{\states}:
\Delta u_{a,a'}^\top \post
\ge
\frac{\bias-1}{\bias}\Delta u_{a,a'}^\top \prior=b_{a,a'}(\bias),
\quad \forall a'\neq a
\right\}.
\]
Here $\Delta u_{a,a'} \in \mathbb{R}^{|\states|}$
is defined component-wise by
$
\bigl(\Delta u_{a,a'}\bigr)(\omega)=u_R(a,\omega)-u_R(a',\omega).
$
Each inequality is a weak incentive-compatibility (IC) constraint, ensuring that $\bias$-biased receiver weakly prefers $a$ to $a'$.
For an interval $J=[\underbias,\upbias]$, define the interval-safe region, i.e., the set of Bayesian posteriors under which action $a$ remains weakly IC for all $\bias\in J$:
\begin{equation}
    \areg{a}{J}=
    \bigcap_{\bias\in J}\areg{a}{\bias}
    =
    \Bigl\{
    \post\in\simplex{\states}:
    \Delta u_{a,a'}^\top\post\ge b_{a,a'}(J)=\max\cbr{b_{a,a'}(\underbias), b_{a,a'}(\upbias)},\ \forall a'\neq a
    \Bigr\}
    \label{eq:interval-safe}
\end{equation}
The last equality follows because $b_{a,a'}(\bias)$ is monotone in $\bias$, with derivative $\Delta u_{a,a'}^\top\prior/\bias^2$.  Thus the binding endpoint can differ across constraints: if $\Delta u_{a,a'}^\top\prior>0$, the safe constraint is determined by $\upbias$, while if $\Delta u_{a,a'}^\top\prior<0$, it is determined by $\underbias$.  This is the main geometric difference from the binary case.  In \Cref{fig:moving-action-regions}, for example, the default-action region $\areg{\defaulta}{\bias}$ is smaller in the high-bias panel, while the region for $a_2$ moves in the opposite direction.

\begin{figure}[ht]
    \centering
    \begin{minipage}[t]{0.45\linewidth}
        \centering
        \includegraphics[width=\linewidth]{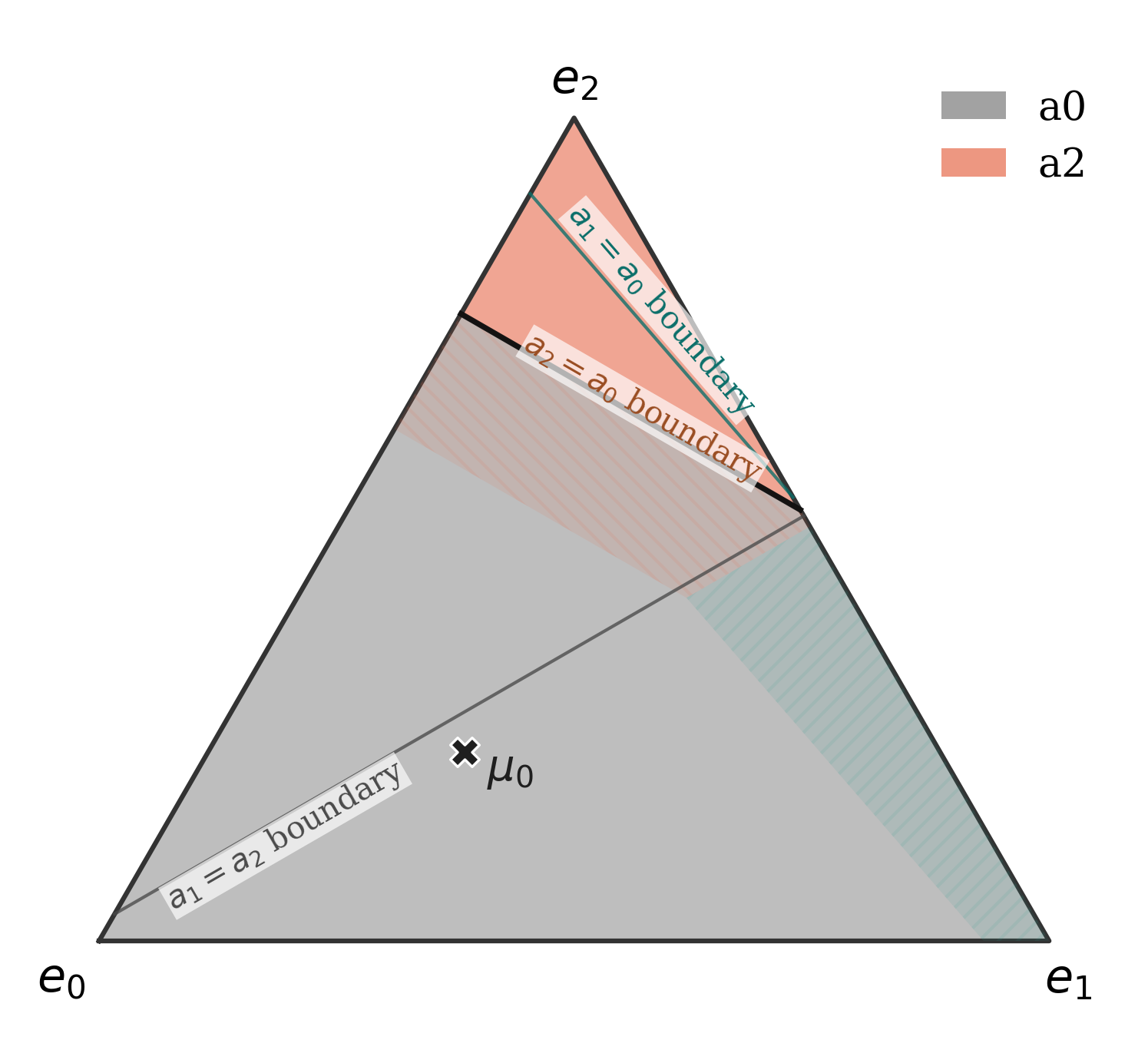}
        {\small (a) $\bias_{\mathrm{low}}=0.55$: $R^{\bias_{\mathrm{low}}}_{a_1}=\emptyset$}
    \end{minipage}
    \hfill
    \begin{minipage}[t]{0.45\linewidth}
        \centering
        \includegraphics[width=\linewidth]{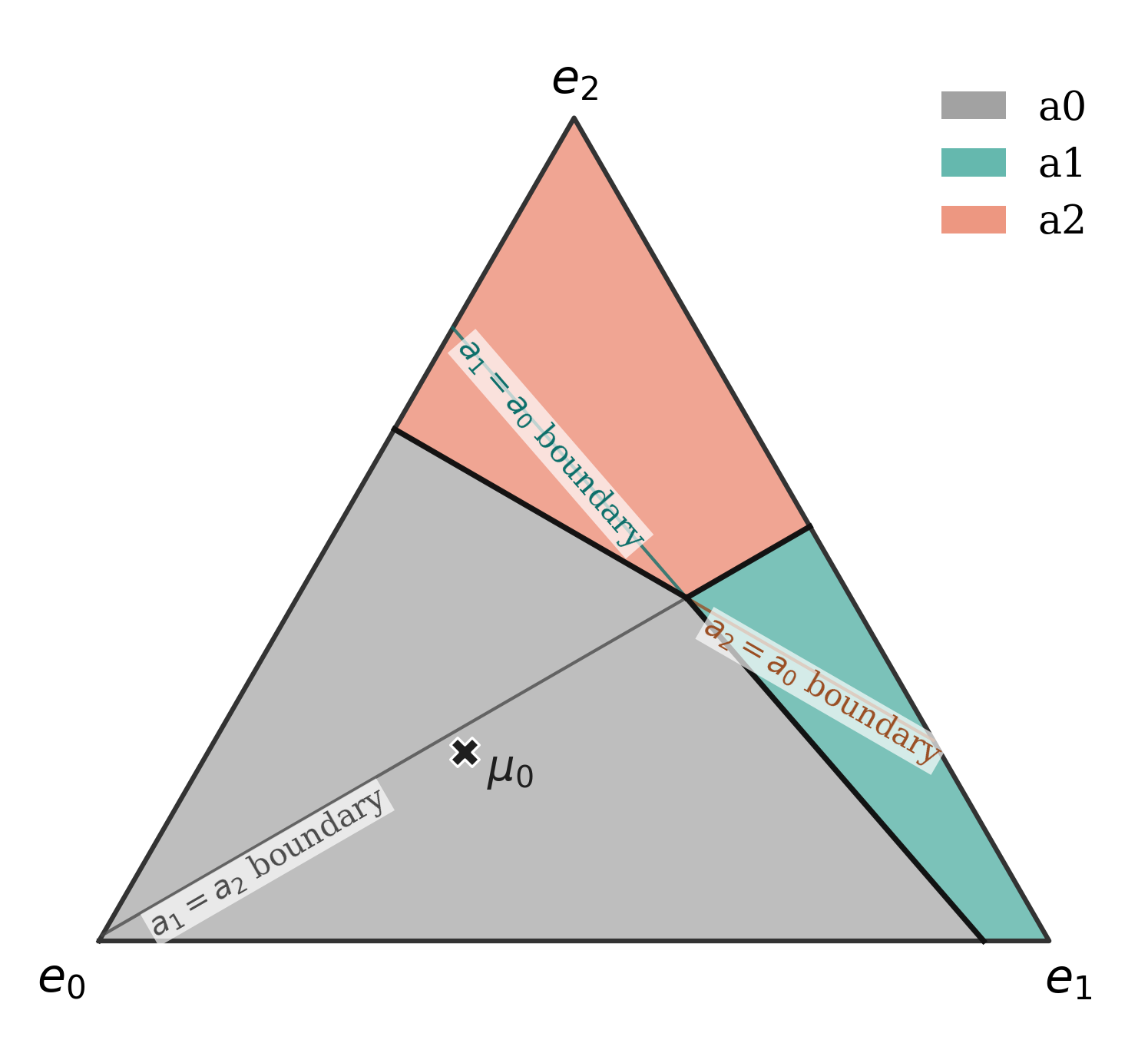}
        {\small (b) $\bias_{\mathrm{high}}=0.85$: $R^{\bias_{\mathrm{high}}}_{a_1}\neq\emptyset$}
    \end{minipage}
    \caption{
  Action regions for $\prior=(0.50,0.27,0.23)$.
In state order $(e_0,e_1,e_2)$,
$u_R(a_0,\cdot)=(0,0,0)$,
$u_R(a_1,\cdot)=(-2.1,0.3,0.9)$, and
$u_R(a_2,\cdot)=(-2.1,-0.3,1.5)$.}
    \label{fig:moving-action-regions}
\end{figure}


We now specify the condition for tie-breaking rules. Define the set of relevant actions:
\[
\actions_{\mathrm{rel}}
=
\Bigl\{a\in\actions:
\exists \post\in\simplex{\states}\text{ that }a\text{ uniquely maximizes }
\sum_{\omega\in\states}(\tbias\post(\omega)+(1-\tbias)\prior(\omega))u_R(a,\omega)
\Bigr\}.
\]
Actions outside $\actions_{\rm rel}$ are weakly dominated for the receiver: they are never unique best responses under the true bias.
We impose the following relevant-action tie-breaking convention: the receiver's tie-breaking rule never selects any action in $\actions\setminus\actions_{\rm rel}$.  This restriction can be understood as a robustness requirement on the regret benchmark.  If tie-only actions were allowed, the full-information benchmark could target a posterior lying exactly on the corresponding indifference set, while a learner who does not perfectly know $\tbias$ could not robustly induce the same action.  
In \Cref{app:relevant_action}, we show that without this restriction, low regret cannot be guaranteed.
Also, note that although $\actions_{\rm rel}$ depends on the unknown \(\tbias\), the sender does not need to know this specific set as \Cref{algo:G-SETC} is run over the original action set \(\actions\). With additional strict-feasibility check defined in \Cref{sec:interval-safe}, the sender can identify $\actions_{\rm rel}$ without knowing $\tbias$. 
Among relevant actions, the receiver breaks ties in favor of the sender, as a standard rule in Bayesian persuasion \cite{kamenica2011bayesian,dughmi2016algorithmic}.

We also keep one standard assumption from \citet{chen2024bias} as below.
\begin{assumption}[Unique Default Action \cite{chen2024bias}]
\label{ass:strict_default}
    There is a unique action $\defaulta$ that maximizes the expected receiver utility based on the prior $\prior:|\arg\max_{a\in \actions}\{\sum_{\omega\in \states }\prior(\omega)u_R(a,\omega)\}|=1$.
\end{assumption}

We propose the \emph{General Safe Exploration (GSE)} algorithm.
The algorithm proceeds in three stages, corresponding to pure exploration, a balanced exploration--exploitation stage, and pure exploitation.
First, the localization stage quickly shrinks the bias interval to length $O(1/\log T)$ so that interval-safe regions are not overly conservative.  
Second, the safe exploration stage follows the same phase structure as \Cref{alg:astc}, to shrink the interval quadratically until its length reaches $O(1/T)$.  
Third, the commitment stage plays an interval-safe scheme.
The resulting regret is $O(\log\log T)$ in total.


\begin{algorithm}[h]
\caption{General Safe Exploration (GSE)}
\label{algo:G-SETC}
\begin{algorithmic}[1]
\Require Horizon $T$, general instance $\cI=(\states,\actions,\prior,u_S,u_R)$
\State Initialize the bias interval $J\leftarrow[\bias_{\min}(\cI),1]$.
\State \textbf{Stage 1: Localization}
\While{$|J|>1/\log T$}
\State Run \textsc{ThresholdTest}$(J,m)$ where $m$ is the midpoint of $J$ and update $J$. \Comment{via \Cref{alg:threshold_test}}
\EndWhile
\State \textbf{Stage 2: Safe exploration}
\State Set $J_0\gets J, r\gets 0$.
\While{$|J_r|>T^{-1}$}
    \State Run \textsc{SafeExplore}$(J_r)$ to obtain $(J_{r+1},\tau^{J_r})$. \Comment{via \Cref{alg:sefi-short}}
    \If{$J_{r+1}=J_r$}
        \State break \Comment{The current safe scheme is already optimal}
    \EndIf
    \State $r\gets r+1$
\EndWhile
\State \textbf{Stage 3: Commitment}
\State Compute and commit to the current interval-safe scheme for all remaining rounds.
\end{algorithmic}
\end{algorithm}

\begin{theorem}[General safe exploration]\label{thm:generalbound}
For every instance \(\cI=(\states,\actions,\prior,u_S,u_R)\) satisfying \Cref{ass:strict_default} and every $\tbias\in(0,1]$, \Cref{algo:G-SETC} runs in polynomial time with respect to \(\cI\) and \(T\). Moreover, $\ireg{\cI}{T}{\Pi^{\mathrm{GSE}}}{\tbias}=O(\log\log T)$ holds with a hidden constant \(\Cgse(\cI,\tbias)\).
\end{theorem}

\subsection{Localization Stage}
\label{sec:localization}
\Cref{fig:moving-action-regions} shows why starting safe exploration on a wide interval can be overly conservative. The interval-safe region $\areg{a_1}{J}$ over a wide interval $J$ is empty, even though $a_1$ may be relevant and utility-improving at the true high bias.
The localization stage first shrinks the interval so that interval safety does not discard such actions. Specifically, the goal in this stage is localizing \(\tbias\) to an \(O(1/\log T)\)-length interval as quickly as possible, without caring about utility loss.
We use the threshold test of \citet{chen2024bias} (see \Cref{alg:threshold_test}): for each candidate bias threshold \(\beta\), it solves LP (see \eqref{eq:direct-lp}) to construct the signaling scheme, where most posteriors lie on the default/non-default indifference boundary at \(\beta\), so the realized action reveals whether \(\tbias\ge \beta\) or \(\tbias<\beta\).
We show that each test is informative with constant probability, so binary search reaches length \(O(1/\log T)\) after \(O(\log\log T)\) informative comparisons and incurs \(O(\log\log T)\) total regret.
Details are deferred to \Cref{app:localization}.

\subsection{Safe Exploration Stage}
\label{sec:safe_exploration}
This stage adapts the safe exploration principle of \Cref{alg:astc} to the general moving-polytope geometry. We give a short version in \Cref{alg:sefi-short} and defer the full implementation to \Cref{alg:sefi}.
First, \Cref{sec:interval-safe} explains how, on the
current uncertainty interval $J$, we compute an interval-safe vertex-supported scheme, i.e., a Bayes-plausible utility-maximizing scheme supported on vertices of the safe polytopes $\{\areg{a}{J}\}_{a\in \actions_{\rm rel}}$, so that every support posterior remains IC for all $\bias\in J$.
Second, \Cref{sec:movingconstraint} explains how vertex posteriors of this scheme are converted into local probes by perturbing movable
binding IC constraints. 
Similar to \Cref{alg:astc}, each informative realization either safely advances an endpoint of $J$ by $\eta=|J|^2$ or localizes $\tbias$ to the adjacent length-$\eta$ interval. Thus each phase reduces $|J|$ to at most $|J|^2$, and \Cref{sec:samplecomplexity} bounds waiting time and total regret.

\begin{algorithm}[ht]
\caption{Fixed-Interval Safe Exploration: Short Version}
\label{alg:sefi-short}
\begin{algorithmic}[1]
\Require Current interval $J=[\underbias,\upbias]$, instance
$\cI=(\states,\actions,\prior,u_S,u_R)$.

\State $L\gets \abr{J}$, $\eta\gets L^2$, $(\ell,r)\gets(\underbias,\upbias)$
\State $\scheme_{\mathrm{vtx}}^J \gets \textsc{VertexSafeScheme}(J,\cI)$
\If{$\up(\scheme_{\mathrm{vtx}}^J)=0$}
    \State \Return $(J,\scheme_{\mathrm{vtx}}^J)$
\EndIf

\For{each informative index $i\in \movv(\scheme_{\mathrm{vtx}}^J)$}
    \State $\mathsf{ori}_i\gets
    \textsc{ChoosemovableBindingConstraint}(i,\scheme_{\mathrm{vtx}}^J,J)$
\EndFor

\State $\scheme^{\mathrm{last}}\gets \scheme_{\mathrm{vtx}}^J$
\While{$r-\ell>\eta$}
    \State $\scheme^{\mathrm{last}}\gets
    \textsc{BuildProbe}
    \rbr{\scheme_{\mathrm{vtx}}^J,
    \cbr{\mathsf{ori}_i}_{i\in \movv(\scheme_{\mathrm{vtx}}^J)},
    \ell,r,\eta}$
    \State $(i,a_t)\gets
    \textsc{RunUntilInformative}\rbr{\scheme^{\mathrm{last}},
    \movv(\scheme_{\mathrm{vtx}}^J)}$
    \State $(\ell,r,\mathsf{status})\gets
    \textsc{RefineInterval}(i,a_t,\mathsf{ori}_i,\ell,r,\eta)$
\EndWhile

\State \Return $([\ell,r],\scheme^{\mathrm{last}})$
\end{algorithmic}
\end{algorithm}

\subsubsection{Interval-safe vertex-supported schemes}
\label{sec:interval-safe}
The interval-safe regions $\{\areg{a}{J}\}_{a\in\actions}$ were defined in \Cref{eq:interval-safe}.  For the current localized interval $J=[\underbias,\upbias]$, also define the strict interval-safe region $\operatorname{int}_{\mathrm{IC}}(\areg{a}{J})=
    \{
        \post\in\simplex{\states}:
        \Delta u_{a,a'}^\top\post>b_{a,a'}(J),
        \ \forall a'\neq a
    \}.$
We additionally remove all actions with
\(\operatorname{int}_{\mathrm{IC}}(\areg{a}{J})=\emptyset\).
This feasibility check exactly recovers the relevant action set \(A_{\rm rel}\), guaranteed once \(J\) is sufficiently small after localization, see \Cref{lemma:relevant-action} in \Cref{app:critial-value}.
Note that this is an action-level filter designed to exclude $A\setminus A_{\rm rel}$. For $a\in A_{\rm rel}$, we still consider the closed polytope \(R_a^J\), as the boundary cases are addressed by the sender-favor tie-breaking rule.

View only posteriors in the interval-safe regions
\(\bigcup_{a\in A_{\rm rel}} \areg{a}{J}\) as feasible.
Using the corresponding tighter IC constraints, we solve a LP in \eqref{eq:safe-lp} to obtain an interval-safe scheme, i.e., a Bayes-plausible utility-maximizing scheme supported in the interval-safe regions \(\bigcup_{a\in A_{\rm rel}} \areg{a}{J}\).

Unlike the binary case, where the optimal scheme $\scheme(m)$ 
places its high posterior on the boundary of the interval-safe region \(R_1^J\) and is therefore easy to probe, an arbitrary interval-safe optimizer in the general case may obscure which posteriors are informative for probing. We therefore replace it by an equivalent vertex-supported scheme. 
Specifically, since each \(\areg{a}{J}\) is a polytope with finitely many vertices, each posterior in the interval-safe scheme can be decomposed into a distribution over vertices of the corresponding polytope. 
Such decomposition is without loss as it preserves utility and Bayes plausibility.
Applying this to all posteriors gives an interval-safe vertex-supported scheme \(\scheme_{\mathrm{vtx}}^J\) on \(J\), see \Cref{app:safeexplore-alg} for details. 
When \(J\) is small, its utility gap from optimal \(\scheme_{\mathrm{opt}}\) is also small.
\begin{proposition}[Safe--optimal gap]
\label{lemma:safe-optimal}
For every interval $J\ni\tbias$,
\(
U_S(\scheme_{\mathrm{opt}})-U_S(\scheme_{\mathrm{vtx}}^J)=O(|J|).
\)
\end{proposition}
This vertex view is the higher-dimensional analogue of the binary safe scheme \(\scheme(\underbias)\) where one posterior lies on the moving boundary of \(R_1^J\). Here informative vertices are similarly identified by binding IC constraints, whose movement in \(\bias\) gives the local one-dimensional probes below.

\subsubsection{Probing the bias via movable binding constraints}
\label{sec:movingconstraint}
Based on the current interval-safe vertex-supported scheme
$\scheme^J_{\rm vtx}$, we construct the following probing scheme.
Call a vertex posterior $\post\in \areg{a}{J}$ informative if it has a movable binding IC constraint: $\exists a'\neq a$,
$\Delta u_{a,a'}^\top \post=b_{a,a'}(J)
=b_{a,a'}(\beta)$ with
$\beta\in\{\underbias,\upbias\}$ where the right-hand side varies with $\bias$, i.e., $\Delta u_{a_i,a_i'}^\top\mu_0\neq 0$. 
Such a constraint plays the role of the binary cutoff.
To probe a candidate value $m \in J$, we move
$\post$ slightly across this one boundary, while leaving the other IC constraints satisfied. Denote the resulting posterior by $\post^{\mathrm{pr}}(m)$. The formal construction are deferred to \Cref{app:safeexplore-alg}. 
This local probe has exactly the same threshold interpretation as in the binary case. If the selected binding constraint is determined by the lower endpoint $\beta=\underbias$,  then \[
    \bestresp{\post^{\mathrm{pr}}(m)}{\tbias}=a
    \quad\Longleftrightarrow \quad
    \tbias\ge m,
    \qquad \bestresp{\post^{\mathrm{pr}}(m)}{\tbias}\neq a
    \quad\Longleftrightarrow\quad
    \tbias< m
\]
If it binds at the upper endpoint, $\beta=\upbias$, the inequalities reverse:
\[
    \bestresp{\post^{\mathrm{pr}}(m)}{\tbias}=a
    \quad\Longleftrightarrow\quad
    \tbias\le m,\qquad \bestresp{\post^{\mathrm{pr}}(m)}{\tbias}\neq a
    \quad\Longleftrightarrow\quad
    \tbias> m
\]
Thus each informative vertex reduces the general problem to a local binary comparison.

We construct the probing scheme $\scheme_{\mathrm{probe}}^{J,\eta}(\ell,r)$ by perturb all informative vertices of $\scheme_{\mathrm{vtx}}^J$
simultaneously. Maintaining a current subinterval $(\ell,r)\subseteq J$ and step size $\eta=|J|^2$, a lower-endpoint vertex probes $m=\ell+\eta$, while an upper-endpoint vertex probes $m=r-\eta$. When an informative signal is realized, the observed action either moves the corresponding endpoint inward by $\eta$, or localizes $\alpha^\star$ to the adjacent length-$\eta$ interval. This is the
same safe-side logic as in \Cref{alg:astc}, except that different vertices may probe different sides of the interval.

We then bound the utility gap between $\scheme_{\mathrm{probe}}^{J,\eta}(\ell,r)$ and $\scheme_{\mathrm{vtx}}^J$. Note that the following property only addresses safe probes, and the wrong-action loss caused by unsafe probes is handled separately later.
\begin{proposition}[Safe--probe gap]
\label{lemma:safe-probe}
For every $(l,r)$,
\(
    U_S(\scheme_{\mathrm{vtx}}^J)
    -
    U_S(\scheme_{\mathrm{probe}}^{J,\eta}(\ell,r))
=O(\up(\scheme_{\mathrm{vtx}}^J)|J|).
\)
\end{proposition}

\subsubsection{A sample complexity dichotomy}
\label{sec:samplecomplexity}
\Cref{lemma:safe-optimal,lemma:safe-probe} control the per-round loss of a fixed probe. The only remaining issue is the waiting time for informative signals. Define the informative probability
\(\up\rbr{\scheme_{\mathrm{vtx}}^J}\) of an interval-safe
vertex-supported scheme \(\scheme_{\mathrm{vtx}}^J\) as the total
probability mass on informative vertex posteriors.
The following dichotomy rules out slow learning. 

\begin{proposition}
\label{prop:samplecomplexity}
There exists an instance-dependent constant $\eps_{\rm dich}(\cI)>0$, such that for every interval $J\ni\tbias$ with $|J|\le\eps_{\rm dich}(\cI)$, one of the following holds:
\begin{enumerate}
    \item $\exists$ constant $c>0$, every interval-safe vertex-supported optimizer satisfies $P_{\rm info}(\scheme^J_{\rm vtx})\ge c$.
    \item There exists an optimal scheme $\scheme_{\mathrm{opt}}$ with
    $P_{\rm info}(\scheme_{\mathrm{opt}})=0$.
\end{enumerate}
\end{proposition}
After the localization stage, $|J|=O(1/\log T)\le\eps_{\rm dich}(\cI)$ for all sufficiently large $T$. 
In the first case, informative signals arrive after $O(1)$ rounds in expectation. A phase with interval $J$ and step size $\eta=|J|^2$ has at most $O(|J|/\eta)=O(1/|J|)$ safe probes and one unsafe probe. Given \Cref{lemma:safe-optimal,lemma:safe-probe}, its expected regret is $O(1)$. 
In the second case, the optimal scheme is also safe within $J$, then there is no safe-optimal gap, only the safe-probe gap $O(P_{\rm info}|J|)$ remains. After multiplying by the $O(1/P_{\rm info})$ waiting time, each safe probe costs $O(|J|)$, hence the whole phase again costs $O(1)$.
Finally, since every completed phase satisfies $|J_{r+1}|\le |J_r|^2$, 
only $O(\log\log T)$ phases are needed.
\begin{proposition}
\label{prop:safe_exploration_regret}
    The total expected regret during the safe exploration stage is $O(\log\log T)$.
\end{proposition}

\subsection{Commitment Stage}
After the safe-exploration stage, the sender commits to the current interval-safe vertex-supported optimizer \(\tau^J_{\mathrm{vtx}}\). If the final interval satisfies \(|J|\le T^{-1}\), \Cref{lemma:safe-optimal} implies that the per-round regret from the full-information benchmark is \(O(T^{-1})\), so the total commitment regret is \(O(1)\). If instead the safe-exploration routine stops earlier because \(P_{\mathrm{info}}(\tau^J_{\mathrm{vtx}})=0\), then by the second branch of \Cref{prop:samplecomplexity}, committing to it incurs no additional regret. Combining the \(O(\log\log T)\) localization cost, the \(O(\log\log T)\) safe-exploration cost, and this \(O(1)\) commitment cost proves \Cref{thm:generalbound}.

\section{Extensions and Open Problems}
\label{sec:conclusion}
Our analysis also clarifies how the safe-exploration principle extends beyond the baseline model. First, under the linear distortion model, persuasion does not harm receiver utility relative to no persuasion (see \Cref{sec_app:receiver-extension}). 
Second, in contextual environments where priors and utilities vary across rounds, the same interval-refinement idea gives \(O(\log\log T)\) regret under a uniform regularity assumption, which holds naturally with a finite contextual family (see \Cref{sec_app:changing_prior}).
Finally, when both the prior and the bias are unknown, the binary model still admits a one-dimensional implementable-threshold representation and hence an \(O(\log\log T)\) regret algorithm, and the problem inherits an \(\Omega(\log T)\) lower bound from unknown-prior persuasion \cite{Li_Lin_2025}, but the matching upper bound in the general case remains open (see \Cref{sec_app:prior-bias-extension}).

Future directions include establishing a matching \(O(\log T)\) upper bound for the general jointly unknown prior-and-bias setting, extending safe exploration to other biased-updating models, and understanding how the guarantees change under noisy or partial action feedback.

\newpage
\bibliographystyle{plainnat}
\bibliography{reference}
\newpage
\appendix
\section{Related Work (Detailed Discussion)}
\label{app:related-work}
A large body of empirical and theoretical work shows that Bayesian updating is a useful normative benchmark but often fails as a descriptive model of human responses to information. Early evidence on conservative updating finds that people revise beliefs in the Bayesian direction but insufficiently, suggesting that new information is only partially incorporated into posterior beliefs \cite{edwards1968conservatism}. Related experimental work documents distorted uses of priors and likelihoods, including base-rate neglect \cite{bar1980base}, representativeness-based inference \cite{grether1980bayes}, and broader errors in probabilistic reasoning \cite{benjamin2019errors}. More recent work develops tests of belief updating and documents systematic over- and underreaction to information \cite{ba2025over}. 
A separate line of work, mostly in political belief updating, studies motivated reasoning and belief polarization, showing that belief revisions may depend on prior commitments rather than only on statistical evidence \cite{taber2006motivated,tappin2020bayesian}. Complementing these findings, recent work in security games explicitly relaxes the assumption that adversaries interpret signals as perfectly Bayesian decision makers: \cite{cooney2019learning} combines machine learning with cognitive models to account for boundedly rational responses to warnings, and \cite{cranford2020toward} uses Instance-Based Learning Theory to model attackers who rely on past experience when interpreting deceptive signals. Together, these studies motivate a tractable model of receivers whose responses to information may be systematically attenuated or distorted rather than fully Bayesian.

Classical Bayesian persuasion studies how an informed sender strategically discloses information to influence a Bayesian receiver's action. The posterior-based formulation used in our paper follows the comparison-of-experiments tradition of Blackwell \cite{blackwell1953equivalent} and the splitting/concavification ideas developed in repeated games with incomplete information \cite{aumann1995repeated}. Kamenica and Gentzkow \cite{kamenica2011bayesian} formulate the canonical Bayesian-persuasion model, where the sender chooses a Bayes-plausible distribution over posteriors. Related early disclosure models include optimal information disclosure \cite{rayo2010optimal}. Subsequent work gives alternative characterizations and broader information-design formulations: Gentzkow and Kamenica \cite{gentzkow2016rothschild} use a Rothschild--Stiglitz perspective, and Bergemann and Morris \cite{bergemann2016information} connect information design to Bayes correlated equilibrium. A parallel algorithmic literature studies the computational complexity of finding optimal schemes when the primitives are known \cite{dughmi2016algorithmic,dughmi2017algorithmic}.


Our work builds on the literature that relaxes this Bayesian-rationality assumption, is closest to work on persuasion with biased or boundedly rational receivers. \cite{de2022non,kobayashi2025dynamic} study static and dynamic persuasion problems in which the receiver follows known non-Bayesian updating rules and the sender designs optimal information accordingly. Relatedly, \cite{tang2021bayesian} provides experimental evidence that behavior in information-design environments can deviate from the Bayesian benchmark and be approximated by linear updating, while \cite{feng2024rationality} studies robust information design when receivers follow quantal-response or approximately rational behavior. 
However, papers above either take the receiver's non-Bayesian rule or bias level as known when designing signals, or focus on a static design problem. The main exception is \cite{chen2024bias}, which also considers a linear belief distortion, but the objective is diagnostic: to infer whether the receiver's bias is above or below a threshold. In contrast, our sender must learn the unknown bias while simultaneously maximizing persuasion utility.

Our work also relates to the growing literature on learning in information design. One strand studies online Bayesian persuasion when the sender does not know the receiver's utility function or type, and designs no-regret signaling policies relative to a benchmark that knows these primitives \cite{castiglioni2020online,castiglioni2021multi,feng2022online,harris2023algorithmic}. Another strand focuses on uncertainty about the receiver's prior or the state distribution \cite{zu2021learning,Li_Lin_2025}. Related work also studies Markovian information-design problems in which the sender learns transition or event probabilities while maintaining persuasiveness \cite{wu2022sequential,bacchiocchi2024markov}. These papers study learning about utility, type, prior, state distribution, or dynamic environment primitives, typically while maintaining Bayesian receiver behavior. By contrast, in our model, the prior and utility functions are known, but the receiver's responsiveness to information is unknown. 

\section{Missing Details in Section \ref{sec:binary}}
\subsection{Constants and Notation}
\label{app:binary-constants}

Fix a binary instance \(\cB=(\prior,\qth)\), where \(\prior\in(0,1)\) and
\(\qth\in(\prior,1)\). Recall
\[
\alpha_{\min}(\cB):=\frac{\qth-\prior}{1-\prior},
\qquad
\nu_{\cB}(\bias):=\prior+\frac{\qth-\prior}{\bias},
\qquad
W_{\cB}(\bias):=\frac{\prior}{\nu_{\cB}(\bias)}.
\]
For every \(\bias\in[\alpha_{\min}(\cB),1]\), we have
\[
\nu_{\cB}(\bias)\in[\qth,1].
\]
Define the two Lipschitz constants
\[
\Lnu(\cB)
:=\sup_{\bias\in[\alpha_{\min}(\cB),1]}
\left|\frac{d}{d\bias}\nu_{\cB}(\bias)\right|
=\frac{\qth-\prior}{\alpha_{\min}(\cB)^2}
=\frac{(1-\prior)^2}{\qth-\prior},
\]
and
\[
L_W(\cB)
:=\sup_{\bias\in[\alpha_{\min}(\cB),1]}
\left|\frac{d}{d\bias}W_{\cB}(\bias)\right|
=\frac{\prior(1-\prior)^2}{\qth-\prior}.
\]
The constants used in \Cref{prop:binary-bsearch,prop:binary-regret-new,thm:lowerbound} can be taken as
\[
\Cbs(\cB)
:=
\bigl(1-\alpha_{\min}(\cB)\bigr)
\left(
\frac{\Lnu(\cB)}{\qth}+L_W(\cB)
\right),
\]
\[
\Cse(\cB)
:=
2\frac{\Lnu(\cB)}{\qth}+1,
\]
and
\[
\Clb(\cB)
:=
\frac13\min\left\{
\frac{\prior^2}{2},\frac{1-\qth}{16}
\right\}.
\]
All these constants depend only on \(\cB\) and are independent of \(T\) and \(\tbias\).

\subsection{Missing Details in Binary Search}
\label{app:binary-baseline}
Recall that $\bias_{\min}(\cB)=(\qth-\prior)/(1-\prior)$, and for any \(\bias\in[\bias_{\min}(\cB),1]\)
\[
\scheme(\bias)
=
\left\{
\left(1-\frac{\prior}{\post_{\cB}(\bias)},0\right),
\left(\frac{\prior}{\post_{\cB}(\bias)},\post_{\cB}(\bias)\right)
\right\}.
\]

\begin{algorithm}[H]
\caption{Binary Search (BS)}
\label{alg:binary-bsearch}
\begin{algorithmic}[1]
\Require Horizon \(T\), binary instance $\cB=(\prior,\qth)$
\State \((\underbias_0,\upbias_0)\gets (\bias_{\min}(\cB),1)\), \(\; k\gets 0\), \(\; t\gets 1\), \(\; M\gets \ceil{2\log_2 T}\)

\Statex \textbf{Phase 1: Exploration}
\While{\(k<M\) and \(t\le T\)}
    \State \(m\gets (\underbias_k+\upbias_k)/2\)
    \State commit to and play \(\scheme(m)\)
    \If{the realized posterior is not \(0\)}
        \If{the receiver chooses action \(1\)}
            \State \((\underbias_{k+1},\upbias_{k+1})\gets (m,\upbias_k)\)
        \Else
            \State \((\underbias_{k+1},\upbias_{k+1})\gets (\underbias_k,m)\)
        \EndIf
        \State \(k\gets k+1\)
    \EndIf
    \State \(t\gets t+1\)
\EndWhile

\Statex \textbf{Phase 2: Commitment}
\State Commit to \(\scheme(\underbias_k)\) for all remaining rounds.
\end{algorithmic}
\end{algorithm}

\begin{proof}[Proof of Proposition~\ref{prop:binary-bsearch}]
If the horizon ends before the algorithm collects \(M\) informative realizations, then the final
epoch is truncated; this can only decrease the regret. It therefore suffices to analyze the
hypothetical process in which the first \(M\) informative epochs are completed, and the policy
then commits.

Let
\[
\Delta_k:=\upbias_k-\underbias_k.
\]
By construction,
\[
\Delta_k=(1-\bias_{\min}(\cB))2^{-k},
\qquad
k=0,1,\dots,M.
\]
For each \(k=1,\dots,M\), define
\[
m_k:=\frac{\underbias_{k-1}+\upbias_{k-1}}{2}.
\]
Epoch \(k\) consists of all rounds in which the algorithm keeps probing \(m_k\) before the
\(k\)-th informative realization occurs. In each such round, the high posterior \(\post_\cB(m_k)\)
appears with probability \(\prior/\post_\cB(m_k)\), so the epoch length \(N_k\) is geometric with
\[
\E[N_k]=\frac{\post_\cB(m_k)}{\prior}.
\]

We first bound the exploration regret. Fix an epoch \(k\).

If \(m_k\le \tbias\), then \(\post_\cB(m_k)\ge \post^\star\), so the high posterior induces action $1$.
Along every sample path of this epoch, the sender obtains utility \(0\) in the first \(N_k-1\)
rounds and utility \(1\) in the last round, while the benchmark obtains
\(\prior/\post^\star\) in every round. Hence the total regret of epoch \(k\) is
\[
N_k\frac{\prior}{\post^\star}-1.
\]
Taking expectations gives
\[
\E[\mathrm{regret\ in\ epoch\ }k]
=
\E[N_k]\frac{\prior}{\post^\star}-1
=
\frac{\post_\cB(m_k)}{\post^\star}-1
=
\frac{\post_\cB(m_k)-\post^\star}{\post^\star}.
\]
Moreover,
\[
\left|\frac{d}{d\bias}\post_\cB(\bias)\right|
=
\frac{\qth-\prior}{\bias^2},
\]
so \(\post_\cB(\bias)\) is \(L_{\post}(\cB)\)-Lipschitz on \([\bias_{\min}(\cB),1]\) with
\[
L_{\post}(\cB):=\frac{\qth-\prior}{\bias_{\min}(\cB)^2}.
\]
Since \(m_k\) is the midpoint of \([\underbias_{k-1},\upbias_{k-1}]\),
\[
|m_k-\tbias|\le \frac{\Delta_{k-1}}{2},
\]
and therefore
\[
\frac{\post_\cB(m_k)-\post^\star}{\post^\star}
\le
\frac{L_{\post}(\cB)|m_k-\tbias|}{\post^\star}
\le
\frac{L_{\post}(\cB)\Delta_{k-1}}{2\post^\star}.
\]
Summing over all such epochs gives
\[
\sum_{k:\,m_k\le \tbias}\E[\mathrm{regret\ in\ epoch\ }k]
\le
\frac{L_{\post}(\cB)}{2\post^\star}\sum_{k=1}^{\infty}\Delta_{k-1}
=
\frac{L_{\post}(\cB)(1-\bias_{\min}(\cB))}{\post^\star}.
\]

If \(m_k>\tbias\), then \(\post_\cB(m_k)<\post^\star\), so even the high posterior induces action $0$.
Along every sample path of this epoch, the sender's utility is \(0\) in all \(N_k\) rounds, while
the benchmark still obtains \(\prior/\post^\star\) per round. Hence the total regret of epoch \(k\) is
\[
N_k\frac{\prior}{\post^\star}.
\]
Taking expectations gives
\[
\E[\mathrm{regret\ in\ epoch\ }k]
=
\E[N_k]\frac{\prior}{\post^\star}
=
\frac{\post_\cB(m_k)}{\post^\star}
\le 1,
\]
where the last inequality uses \(\post_\cB(m_k)\le \post^\star\). There are at most \(M\) such epochs.
Hence
\[
\Reg_T^{\mathrm{explore}}
\le
M+\frac{L_{\post}(\cB)(1-\bias_{\min}(\cB))}{\post^\star}.
\]

We next bound the regret after commitment. The policy commits to \(\scheme(\underbias_M)\). Since
\(\underbias_M\le \tbias\), this scheme is conservative and therefore always persuasive. Also,
\[
0\le \tbias-\underbias_M\le \Delta_M=(1-\bias_{\min}(\cB))2^{-M}.
\]
Define
\[
W(\bias):=\frac{\prior}{\post_\cB(\bias)}
=
\frac{\prior\bias}{\qth-\prior+\prior\bias}.
\]
Then the benchmark utility is \(W(\tbias)\), while the committed scheme yields
\(W(\underbias_M)\). Since
\[
W'(\bias)=\frac{\prior(\qth-\prior)}{(\qth-\prior+\prior\bias)^2},
\]
we have the uniform bound
\[
|W'(\bias)|\le L_W(\cB):=\frac{\prior(1-\prior)^2}{\qth-\prior},
\qquad
\bias\in[\bias_{\min}(\cB),1].
\]
Therefore the total regret in the commitment phase is at most
\[
T\bigl(W(\tbias)-W(\underbias_M)\bigr)
\le
TL_W(\cB)(\tbias-\underbias_M)
\le
TL_W(\cB)(1-\bias_{\min}(\cB))2^{-M}
\le
L_W(\cB)(1-\bias_{\min}(\cB)),
\]
where the last step uses \(T2^{-M}\le 1\).

Combining the exploration and commitment bounds yields
\[
\Reg_T(\Pi^{\mathrm{BS}};\tbias)
\le
\ceil{2\log_2 T}
+
\frac{L_{\post}(\cB)(1-\bias_{\min}(\cB))}{\post^\star}
+
L_W(\cB)(1-\bias_{\min}(\cB)).
\]
\end{proof}

\subsection{Proof of Proposition \ref{prop:binary-regret-new}}
For a probe value $m\in[\bias_{\min}(\cB),1]$, the algorithm uses 
the two-point scheme supported on
\(\{0,\post_\cB(m)\}\):
\[
\scheme(m)
=
\left\{
\left(1-\frac{\prior}{\post_\cB(m)},\,0\right),
\left(\frac{\prior}{\post_\cB(m)},\,\post_\cB(m)\right)
\right\}.
\]
The receiver takes action $1$ if and only if
\[
(1-\tbias)\prior+\tbias\post_\cB(m)\ge \qth
\iff m\le \tbias.
\]
Hence the sender's one-round regret under probe $m$ is
\[
r(m)=
\begin{cases}
\dfrac{\prior}{\post^{\star}}-\dfrac{\prior}{\post_\cB(m)}, & m\le \tbias,\\[6pt]
\dfrac{\prior}{\post^{\star}}, & m>\tbias.
\end{cases}
\]

The algorithm repeats the same probe until the first High signal appears, so each probe corresponds to
one epoch. Let $N(m)$ be the number of physical rounds spent in the epoch for probe $m$. Then
$N(m)\sim \mathrm{Geom}(p_H(m))$ and
\[
\E[N(m)] = \frac{1}{p_H(m)} = \frac{\post_\cB(m)}{\prior}.
\]
Since the probe $m$ is repeated until the first High signal appears, the total regret of one epoch is
\[
N(m)\frac{\prior}{\post^\star}-\mathbf{1}\{m\le\tbias\}.
\]
Therefore,
\[
R_{\mathrm{ep}}(m)
=
\mathbb E\!\left[N(m)\frac{\prior}{\post^\star}-\mathbf{1}\{m\le\tbias\}\right],
\]
which gives
\[
R_{\mathrm{ep}}(m)=
\begin{cases}
\dfrac{\post_\cB(m)-\post^\star}{\post^\star}, & m\le\tbias,\\[6pt]
\dfrac{\post_\cB(m)}{\post^\star}, & m>\tbias.
\end{cases}
\]
In particular, if $m>\tbias$, then $\post_\cB(m)\le \post^{\star}$ and hence
\[
R_{\mathrm{ep}}(m) \le 1.
\]

Next define
\[
C(\cB):=\sup_{x\in[\bias_{\min}(\cB),1]} \left|\frac{d}{dx}\post_\cB(x)\right|
=\sup_{x\in[\bias_{\min}(\cB),1]} \frac{\qth-\prior}{x^2}
=\frac{\qth-\prior}{\bias_{\min}(\cB)^2}.
\]
For every $m\le \tbias$,
\[
0\le \post_\cB(m)-\post^{\star} \le C(\cB)(\tbias-m),
\]
so
\[
R_{\mathrm{ep}}(m)\le \frac{C(\cB)}{\post^{\star}}(\tbias-m).
\]
We now bound the exploration regret phase by phase. In a given phase, let the current interval be
$[a,b]$ with length $L:=b-a$, and let the current step size be $\eps$. The algorithm probes the points
\[
m_j=a+j\eps,
\qquad j=1,\ldots,J,
\]
,then either stops at the first point $m_J$ with high signal inducing action $0$ or exits the
interval and updates to $[m_{J-1},b]$. In both cases,
\[
\tbias-m_j \le (J-j)\eps,
\qquad
J\le \frac{L}{\eps}+1.
\]
There is at most one epoch with high posterior inducing action $0$ in each phase, and its regret is at most $1$.
Therefore the total regret except the above special epoch in this phase is bounded by
\begin{align*}
\sum_{j=1}^{J-1} R_{\mathrm{ep}}(m_j)
&\le \frac{C(\cB)}{\post^{\star}} \sum_{j=1}^{J-1}(\tbias-m_j) \\
&\le \frac{C(\cB)}{\post^{\star}}\,\eps\sum_{r=1}^{J-1} r
\le \frac{C(\cB)}{\post^{\star}}\left(\frac{L^2}{2\eps}+\frac{L}{2}\right).
\end{align*}

It remains to count phases. Let $\eps_0=1/2$ be the initial step size. After each phase, the new
interval length is at most the previous step size, and the next step size is squared. Thus if $L_r$ and
$\eps_r$ denote the interval length and step size at the beginning of phase $r$, then
\[
L_{r+1}\le \eps_r,
\qquad
\eps_{r+1}=\eps_r^2.
\]
Hence
\[
\eps_r = 2^{-2^r},
\qquad
L_r\le \eps_{r-1}\le 1
\quad (r\ge 1),
\]
and in particular $\eps_r \ge L_r^2$ for every $r\ge 1$. Therefore every phase after the first incurs regret at most $C(\cB)/\post^{\star}+1$. The first phase also has
$L_0\le 1$ and $\eps_0=1/2$, so it contributes the same order.

Exploration terminates once the interval length becomes at most $1/T$, which is implied by
$\eps_r\le 1/T$. Since $\eps_r=2^{-2^r}$, the number of phases is at most
\[
P \le 1+\left\lceil \log_2\log_2 T\right\rceil.
\]
Therefore the total exploration regret is
\[
O\!\left(\left(\frac{C(\cB)}{\post^{\star}}+1\right)\log\log T\right).
\]

In the commitment stage, the algorithm commits to the left endpoint $\hat\bias$ of the final interval.
Because $\hat\bias\le \tbias$ and $|\tbias-\hat\bias|\le 1/T$, its one-round regret is
\[
\frac{\prior}{\post^{\star}}-\frac{\prior}{\post_\cB(\hat\bias)}
\le \frac{\post_\cB(\hat\bias)-\post^{\star}}{\post^{\star}}
\le \frac{C}{\post^{\star}}\,|\tbias-\hat\bias|
\le \frac{C}{T\post^{\star}}.
\]
Summing over the remaining $T$ rounds gives commitment regret $O(C(\cB)/\post^{\star})$. Combining the
exploration and commitment parts proves
\[
\Reg_T\!\left(\Pi^{\mathrm{SE}};\tbias\right)
= O\!\left(\left(\frac{2C(\cB)}{\post^{\star}}+1\right)\log\log T\right).
\]
Finally, since
\[
\post^{\star}=\frac{\qth-(1-\tbias)\prior}{\tbias},
\qquad
C(\cB)=\frac{(1-\prior)^2}{\qth-\prior},
\]
we obtain
\[
\frac{C(\cB)}{\post^{\star}}\leq \frac{C(\cB)}{\qth}
= \frac{(1-\prior)^2}{(\qth-\prior)\qth}.
\]

\subsection{Proof of Theorem \ref{thm:lowerbound}}
As discussed above, if $\tbias<\bias_{\min}(\cB)$ then persuasion is impossible and the regret is identically
zero. Hence it suffices to prove the lower bound on the subfamily $\tbias\in[\bias_{\min}(\cB),1]$. Under the
change of variables
\[
\post^{\star}=\prior+\frac{\qth-\prior}{\tbias},
\]
this is equivalent to proving a lower bound for an unknown threshold $\post^{\star}\in[\qth,1]$.

By Yao's principle \cite{yao1977probabilistic}, it is enough to exhibit a distribution over $\post^{\star}$ such that every
deterministic strategy suffers expected regret $\Omega(\log\log T)$. We choose $\post^{\star}$ uniformly on
$[\qth,1]$.

It is without loss of generality to restrict attention to
two-point schemes supported on $\{0,\post\}$ with $\post\in[\qth,1]$. Posteriors below $\prior$ can be
replaced by $0$. This modification strictly improves performance, as posteriors below $\prior$ yield zero utility. Moreover, shifting these values to $0$ minimizes the probability mass required for the lower signals, thereby allowing for a higher probability on the utility-generating posteriors.
For posteriors above $\prior$, denoted by $\overline{\nu}_1 < \overline{\nu}_2 < \cdots$, they can be decomposed 
into a convex combination of binary signal schemes $\{0, \overline{\nu}_1\}, \{0, \overline{\nu}_2\}, \dots$. 
Since we restrict our attention to deterministic algorithms, it suffices to consider a binary signal scheme 
with support $\{0, \nu\}$, which induces the useful posterior $\nu$ with probability $\frac{\prior}{\nu}$. We use High signal to denote this useful posterior $\nu$ in each signal scheme and Low signal to denote the posterior $0$.

We now work in \emph{event time}. Let $\tau_1<\tau_2<\cdots$ be the physical rounds in which a High
signal is realized, and let
\[
N:=\max\{k:\tau_k\le T\}
\]
be the total number of High signals by time $T$. Under a deterministic strategy, the history of High realizations only reveals comparisons of $\post^{\star}$ with queried values, so
conditional on the current history the posterior distribution of $\post^{\star}$ remains uniform on some
feasible interval $I_k=[a_k,b_k]$. Within each phase, the strategy may therefore be viewed as choosing a
deterministic descending sequence of probes until the first High signal that inducing action $0$ arrives.

Following Kleinberg and Leighton~\citep{kleinberg2003value}, define phase $k\ge 0$ to end after either
\[
M_k:=2^{2^k}-1
\]
High signals inducing action $1$ or the first High signal inducing action $0$, whichever comes first. Let $I_k=[a_k,b_k]$
denote the feasible interval at the beginning of phase $k$, and let
\[
E_k:=\left\{|I_k|\ge \frac{1-\qth}{2}\,2^{-2^k}\right\}.
\]
The same feasible-interval argument as in Claims~2.3--2.4 of Kleinberg and
Leighton~\citep{kleinberg2003value} gives
\[
\Prob(E_k)\ge \frac12.
\]
We will show that every phase contributes a constant expected regret.

Fix a phase $k$ and condition on $E_k$. Let
\[
c_k:=\frac{a_k+b_k}{2}
\]
be the midpoint of the feasible interval.

\paragraph{Case 1:} Some probe $\post<c_k$ is used during phase $k$. Since $\post\le 1$, the
corresponding High signal occurs with probability at least $\prior$. Also, conditional on $I_k$, the true
threshold is uniform on $I_k$, so with probability $1/2$ we have $\post^{\star}\ge c_k>\post$.
Whenever both events occur, the realized High signal induces action $0$ and the sender incurs instantaneous
regret at least
\[
\frac{\prior}{\post^{\star}}\ge \prior.
\]
Therefore, conditional on $E_k$, the expected regret in this phase is at least
\[
\frac{\prior}{2}\cdot \prior = \frac{\prior^2}{2}.
\]

\paragraph{Case 2:} Every probe used during phase $k$ satisfies $\post\ge c_k$. Conditional on
$\post^{\star}\le c_k$ (which has probability $1/2$ given $I_k$ is uniform), every realized High signal in
this phase induces action $1$. Hence the phase ends only when it accumulates exactly $M_k$ High
signals.

Let $S_k$ be the number of physical rounds in phase $k$, and for each round $t$ in this phase let
$H_t\in\{0,1\}$ be the indicator of a High realization. If $\mathcal F_t$ is the natural filtration, then
conditional on $\mathcal F_{t-1}$ the strategy's choice $\post_t$ is fixed and
\[
\E[H_t\given \mathcal F_{t-1}] = \frac{\prior}{\post_t} \le \frac{\prior}{c_k}.
\]
Thus $Z_t:=H_t-\E[H_t\given \mathcal F_{t-1}]$ is a bounded martingale difference sequence, and
optional stopping gives
\[
M_k
= \E\!\left[\sum_{t=1}^{S_k} H_t\right]
= \E\!\left[\sum_{t=1}^{S_k} \frac{\prior}{\post_t}\right]
\le \frac{\prior}{c_k}\,\E[S_k].
\]
Hence
\[
\E[S_k] \ge M_k\frac{c_k}{\prior}.
\]
On every physical round of this phase, because $\post_t\ge c_k$, the sender's one-round regret is at
least
\[
\frac{\prior}{\post^{\star}}-\frac{\prior}{c_k}
= \frac{\prior(c_k-\post^{\star})}{c_k\post^{\star}}.
\]
Multiplying by the lower bound on $\E[S_k]$ yields
\[
\E[\text{Regret in phase }k\given \post^{\star}\le c_k, E_k]
\ge M_k\,\E\!\left[\frac{c_k-\post^{\star}}{\post^{\star}} \given \post^{\star}\le c_k, E_k\right]
\ge M_k\,\E[c_k-\post^{\star}\given \post^{\star}\le c_k, E_k],
\]
where the last inequality uses $\post^{\star}\le 1$. Since $\post^{\star}$ is uniform on $[a_k,c_k]$ under the
conditioning,
\[
\E[c_k-\post^{\star}\given \post^{\star}\le c_k, E_k] = \frac{|I_k|}{4}.
\]
Therefore,
\[
\E[\text{Regret in phase }k\given E_k]
\ge \frac12\,M_k\frac{|I_k|}{4}
\ge \frac{1-\qth}{16}\,(1-2^{-2^k})
\ge \frac{1-\qth}{16}\cdot \frac12,
\]
so in particular the phase contributes at least a positive constant.

Combining the two cases, there exists a universal constant
\[
c_0(\cB):=\frac12\min\!\left\{\frac{\prior^2}{2},\frac{1-\qth}{16}\right\}>0
\]
such that every phase contributes at least $c_0(\cB)$ expected regret. Since the number of completed phases
is $\Omega(\log\log N)$, we obtain
\[
\Reg_T\ge c_0(\cB)\,\Omega(\log\log N).
\]

It remains to relate event time to physical time. For each physical round $t$, the chosen probe
satisfies $\post_t\le 1$, hence the High realization probability obeys
\[
\E[H_t\given \mathcal F_{t-1}] = \frac{\prior}{\post_t} \ge \prior.
\]
Therefore $X_t:=H_t-\E[H_t\given \mathcal F_{t-1}]$ is a bounded martingale difference sequence, and
Azuma--Hoeffding implies
\[
\Prob\!\left(N\le \frac{\prior T}{2}\right)
\le \exp\!\left(-\frac{\prior^2T}{8}\right).
\]
Thus with probability at least $1-e^{-\prior^2T/8}$ we have $N\ge \prior T/2$, and consequently
\[
\Reg_T
\ge \left(1-e^{-\prior^2T/8}\right)c_0(\cB)\log_2\log_2\!\left(\frac{\prior T}{2}\right)
= \Omega(\log\log T).
\]
This completes the proof.

\section{Missing Details in Section \ref{sec:general}}
\subsection{Constants and Notation}
\label{app:general-constants}

We collect the constants used in \Cref{sec:general}. All of them are independent of the
horizon \(T\).

Because states with zero prior probability are never reached, \(\prior\) can be viewed as full support on \(\states\) without loss of
generality. Then define
\[
\dprior
=
\min\left\{
\dist\!\left(\prior,\partial \areg{a_0}{1}\right),
\dist\left(\prior,\partial\simplex{\states}\right)
\right\}.\]
Here $\partial$ denote the boundary of one region.
Moreover, $\areg{\defaulta}{1}$ is the smallest default-action region over $\bias\in(0,1]$, so $\dprior$ is a uniform lower bound on the distance from $\prior$ to the nearest default-action region boundary.
By \Cref{ass:strict_default}, $\dprior>0$.
Also, Define \[
U_{\max}=\max_{a\in\actions,\omega\in\states}|u_S(a,\omega)|,
\qquad
\Dumax=2U_{\max}.
\]


\paragraph{Perturbation constants} Fix an action \(a\).  Let \(n:=|\states|\) and \(d:=n-1\).  Let
\[
    H_0:=\{x\in\mathbb R^n:\ones^\top x=0\},
\]
and let \(Q\in\mathbb R^{n\times d}\) have orthonormal columns spanning \(H_0\).  Fix any
\(\bar\post\) with \(\ones^\top\bar\post=1\).  Every vector \(\post\) with
\(\ones^\top\post=1\) can be written uniquely as
\[
    \post=\bar\post+Qz,
    \qquad z\in\mathbb R^d,
\]
and Euclidean distances are preserved:
\[
    \|\post-\post'\|_2=\|z-z'\|_2.
\]

Collect all inequality constraints defining \(\areg{J}{a}\): the IC constraints
\(\Delta u_{a,a'}^\top\post\ge b^J_{a,a'}\) and the simplex nonnegativity constraints
\(\post(\omega)\ge0\).  Stack their normals as rows of a matrix \(A_a\in\mathbb R^{m_a\times n}\)
and write
\[
    \areg{J}{a}
    =
    \{\post\in\mathbb R^n:
        \ones^\top\post=1,\ 
        A_a\post\ge b_a(J)
    \}.
\]
Define
\[
    \widetilde A_a:=A_aQ,
    \qquad
    \widetilde b_a(J):=b_a(J)-A_a\bar\post,
\]
and the reduced polytope
\[
    \widetilde R_a^J
    :=
    \{z\in\mathbb R^d:\widetilde A_a z\ge \widetilde b_a(J)\}.
\]
Then \(\post=\bar\post+Qz\in\areg{J}{a}\) if and only if
\(z\in\widetilde R_a^J\), and
\[
    \dist(\post,\areg{J}{a})
    =
    \dist(z,\widetilde R_a^J).
\]

For \(B\subseteq[m_a]\), let \(\widetilde A_{a,B}\) be the submatrix of \(\widetilde A_a\) with rows in
\(B\).  Define
\[
    \mathcal B_a
    :=
    \left\{
        B\subseteq[m_a]:
        1\le |B|\le d,
        \ \rank(\widetilde A_{a,B})=|B|
    \right\}.
\]
For \(B\in\mathcal B_a\), define the right pseudoinverse
\[
    \widetilde A_{a,B}^{\dagger}
    :=
    \widetilde A_{a,B}^{\top}
    \left(
        \widetilde A_{a,B}\widetilde A_{a,B}^{\top}
    \right)^{-1}.
\]
Finally we can define
\[
    \kappa_a
    :=
    \max_{B\in\mathcal B_a}
    \|\widetilde A_{a,B}^{\dagger}\|_2,
    \qquad
    \kappa:=\max_{a\in\actions}\kappa_a .
\]

\paragraph{Geometric constants}
We then define three new geometric constants:
\[
    G_{\max}=\max_{a\neq a'} |\Delta u_{a,a'}^\top \prior|,
    \qquad
    L_b:=
\max_{a\in\actions}
\sup_{\bias\in[\alpha_{\min}(\cI),1]}
\left\|
\frac{d}{d\bias}\widetilde b_a(\bias)
\right\|_2 .
\]
Both of them are finite and depends only on \(\cI\). 
For $L_b$, using
\(\alpha_{\min}(\cI)\ge\dprior/\sqrt{2}\), it can be upper bounded by an instance-dependent
multiple of \(\Gmax/\dprior^2\).

\paragraph{Regret constants}
The hidden constant in \Cref{cor:localization-regret} is
\[
\Cloc:=\frac{\Dumax}{\pmin}
=\frac{\sqrt{2}\Dumax}{\dprior}.
\]

The posterior-distance constant is
\[
\Kdist:=\kappa L_b .
\]
The Bayes-plausibility repair constant from \Cref{lemma:utility_stability_changed} is
\[
\Crep
:=
U_{\max}\left(\sqrt{|\states|}+\frac{4}{\dprior}\right).
\]
We use
\[
\Ksafe:=\Crep\Kdist,
\qquad
\Kprobe:=\Crep\Kdist
\]
for \Cref{lemma:safe-optimal,lemma:safe-probe}, respectively.

Let \(\cinf\) denote the constant \(c\) in the first branch of
\Cref{prop:samplecomplexity}. One valid phase-level hidden constant is
\[
\Cphase= \max\left\{\frac{2(\Ksafe+\Kprobe+\Dumax)}{c_{\mathrm{info}}},2\Kprobe+\Dumax\right\}.
\]
Thus the general-case constant in \Cref{thm:generalbound} may be chosen as any finite constant
satisfying
\[
\Cgse(\cI,\tbias)
\ge
\Cloc+\Cphase+1 .
\]

\subsection{Missing Details in the Localization Stage}
\label{app:localization}

Our main goal in the first stage is to conduct pure exploration and shrink the uncertainty interval as quickly as possible. This objective is aligned with \citet{chen2024bias}, whose algorithm is designed to test a single threshold with minimum sample complexity. 
Therefore, we adopt the same threshold-testing construction, see \Cref{alg:threshold_test}.
It designs a signaling scheme for threshold $\beta$ by solving a carefully designed LP, 
which constructs a direct-recommendation scheme whose non-default recommendations lie exactly on the receiver's decision boundary between non-default action and default action at $\beta$.
Conditional on such a non-default recommendation being realized, the observed receiver action reveals whether $\tbias\ge \beta$ or $\tbias<\beta$. 
The LP objective maximizes the probability of these informative non-default recommendations, equivalently minimizing the expected waiting time until the comparison is revealed.

\paragraph{The threshold-test LP and algorithm.}
Fix a test point \(\beta\).  We use a direct-recommendation scheme
\(\mech(a\mid\omega)\) by solving the following LP:
\begin{equation}
\begin{aligned}
\text{Maximize}\quad 
& \sum_{a\in \actions\setminus\{\defaulta\}}\ \sum_{\omega\in\states} \pi(a\mid\omega)\,\prior(\omega) \\[0.5ex]
\text{subject to:}\quad 
& \textbf{Optimality of $a$ over other actions:}\ \forall a\in \actions,\ \forall a'\in \actions\setminus\{a\} \\
& \sum_{\omega\in\states}\pi(a\mid\omega)\,\prior(\omega)\Big[\beta\Delta u_{a,a'}(\omega)
+(1-\beta)\sum_{\omega'\in\states}\prior(\omega')\Delta u_{a,a'}(\omega')\Big]\ge 0; \\[0.5ex]
& \textbf{Indifference between $a$ and $\defaulta$:}\ \forall a\in \actions\setminus\{\defaulta\}\\
& \sum_{\omega\in\states}\pi(a\mid\omega)\,\prior(\omega)\Big[\beta\Delta u_{a,\defaulta}(\omega)
+(1-\beta)\sum_{\omega'\in\states}\prior(\omega')\Delta u_{a,\defaulta}(\omega')\Big]= 0; \\[0.5ex]
& \textbf{Probability distribution constraints:}\ \forall \omega\in\states\\
& \sum_{a\in \actions}\pi(a\mid\omega)=1
\quad\text{and}\quad
\forall a\in \actions,\ \pi(a\mid\omega)\ge 0. 
\end{aligned}
\label{eq:direct-lp}
\end{equation}

\begin{algorithm}[H]
\caption{\textsc{ThresholdTest} at candidate bias \(\beta\)}
\label{alg:threshold_test}
\begin{algorithmic}[1]
\Require Current interval \(J=[\underbias,\upbias]\), test point \(\beta\in J\), general instance $\cI=(\states,\actions,\prior,u_S,u_R)$.
\State Solve the threshold-test LP at \(\beta\), and let \(\mech^\beta\) be an optimal solution.
\Repeat
    \State Play \(\mech^\beta\) for one round.
    \State Observe the realized recommendation action \(a_{\mathrm{rec}}\) and the receiver's actual action.
\Until{\(a_{\mathrm{rec}}\neq \defaulta\)}
\If{the receiver takes the default action \(a_0\) under recommendation $a_{\mathrm{rec}}\neq a_0$}
    \State \Return \([\underbias,\beta]\).
\Else
    \State \Return  \([\beta,\upbias]\). 
\EndIf
\end{algorithmic}
\end{algorithm}

Note, however, that as in the binary case, not every candidate bias is testable. When the bias is very small, the receiver may remain too close to the prior for any signal to affect her decision. We establish below the existence of an instance-dependent detectable threshold $\bias_{\min}(\cI)$. Since the absence of persuasion implies zero regret, and one additional test suffices to detect whether persuasion is possible, we focus on $\tbias\in[\bias_{\min}(\cI),1]$ in our analysis.

\begin{proposition}
\label{prop:detectable-localization}
There exists an instance-dependent threshold $\bias_{\min}(\cI)\in[\dprior/\sqrt{2},1]$ such that threshold testing is feasible if and only if $\beta\ge\bias_{\min}(\cI)$. 
Moreover, for every feasible candidate bias, the testing scheme produces an informative realization with probability at least $\dprior/\sqrt{2}$. 
\end{proposition}
\begin{proof}
    For every non-default action \(a\neq \defaulta\), define the indifference hyperplane in distorted-belief
space
\[
    \mathcal I_a
    :=
    \left\{
        \dpost\in\simplex{\states}:
        \Delta u_{a,\defaulta}^\top \dpost =0
    \right\}.
\]
For a candidate bias level \(\beta\in[0,1]\), define its pullback to Bayesian-posterior space:
\[
    \mathcal I_{a,\beta}
    :=
    \left\{
        \post\in\simplex{\states}:
        (1-\beta)\prior+\beta\post
        \in \mathcal I_a
    \right\}.
\]

From Theorem 5.2 in \cite{chen2024bias}, we have that 
For a candidate \(\beta\in[0,1]\), threshold testing between
\(\tbias\ge \beta\) and \(\tbias<\beta\) is feasible if and only if
\(\mathcal I_{a,\beta}\neq\emptyset\) for some \(a\neq\defaulta\). Namely, a threshold
test at level \(\beta\) exists exactly when some non-default action can be placed on an indifference
boundary with the default action at the distorted belief generated by \(\beta\).

We now show that the set of feasible thresholds is an upper interval:
there exists \(\alpha_{\min}(I)\in[\dprior/\sqrt{2},1]\) such that \(\beta\) is testable if and only if
\(\beta\ge \alpha_{\min}(I)\).

 Let
\[
    \mathcal S
    :=
    \left\{
        \beta\in[0,1]:
        \exists a\neq\defaulta
        \text{ such that }
        \mathcal I_{a,\beta}\neq\emptyset
    \right\}.
\]
If \(\beta\in\mathcal S\) and \(\beta'\ge \beta\), choose
\(\post\in\mathcal I_{a,\beta}\).  Define
\[
    \post'
    :=
    \frac{\beta}{\beta'}\post
    +
    \left(1-\frac{\beta}{\beta'}\right)\prior .
\]
Then \(\post'\in\simplex{\states}\), and
\[
    (1-\beta')\prior+\beta'\post'
    =
    (1-\beta)\prior+\beta\post
    \in \mathcal I_a .
\]
Hence \(\post'\in\mathcal I_{a,\beta'}\), so \(\beta'\in\mathcal S\).  Therefore
\(\mathcal S\) is an upper interval, and we may set
\[
    \alpha_{\min}(I):=\inf \mathcal S .
\]
Finally, suppose \(\beta\in\mathcal S\).  Then for some \(a\neq\defaulta\) there exists
\(\post\in\simplex{\states}\) such that
\[
    \dpost:=(1-\beta)\prior+\beta\post\in\mathcal I_a .
\]
Since \(\mathcal I_a\) is contained in the boundary of the Bayesian default-action region, the
definition of \(\dprior\) gives
\[
    \|\dpost-\prior\|_2\ge \dprior .
\]
On the other hand,
\[
    \|\dpost-\prior\|_2
    =
    \beta\|\post-\prior\|_2
    \le
    \beta\sqrt{2},
\]
because both \(\post\) and \(\prior\) lie in the simplex.  Thus
\(\beta\ge \dprior/\sqrt{2}\), and in particular
\(\alpha_{\min}(I)\ge \dprior/\sqrt{2}\).

Now it remains to lower bound the useful-signal probability.  Let \(p_\beta\) denote the optimal objective
value of the threshold-test LP, i.e., the total probability of recommending a non-default action.  Let
\(\post_{\mathrm{use}}\) be the probability-weighted average posterior conditional on a non-default
recommendation, and let \(\post_{\mathrm{def}}\) be the posterior conditional on the default
recommendation.  Bayes plausibility gives
\[
    \prior
    =
    p_\beta \post_{\mathrm{use}}
    +
    (1-p_\beta)\post_{\mathrm{def}} .
\]
Equivalently,
\[
    p_\beta
    =
    \frac{\|\prior-\post_{\mathrm{def}}\|_2}
         {\|\post_{\mathrm{use}}-\post_{\mathrm{def}}\|_2}.
\]
For an optimal threshold-test scheme, the default posterior can be chosen on the boundary of the
Bayesian default-action region.  Otherwise, if \(\post_{\mathrm{def}}\) were strictly inside that region,
moving it farther away from \(\prior\) along the ray passing through \(\post_{\mathrm{def}}\) would
preserve default optimality and would strictly increase the feasible probability of non-default
recommendations, contradicting optimality of the LP objective.  Hence
\[
    \|\prior-\post_{\mathrm{def}}\|_2\ge \dprior .
\]
Since any two points in the simplex are at Euclidean distance at most \(\sqrt{2}\),
\[
    p_\beta
    \ge
    \frac{\dprior}{\sqrt{2}} .
\]
Thus every feasible threshold-test scheme emits a useful non-default recommendation with probability
at least \(\dprior/\sqrt{2}\).
\end{proof}

In this stage, our algorithm stops once the uncertainty interval satisfies $|J|=O(1/\log T)$, which requires $O(\log\log T)$ threshold tests. By \Cref{prop:detectable-localization}, the probability of obtaining an informative realization is bounded below by a constant. Hence, each threshold test has constant sample complexity, and the total regret incurred in this stage is $O(\log\log T)$.

\begin{proposition}
\label{cor:localization-regret}
Given the target number of informative signals in the localization stage is $O(\log\log T)$, the expected regret during the localization stage is $O(\log \log T)$.
\end{proposition}
\begin{proof}
    The targeted number of informative signals are $N=O(\log\log T)$. Let \(X_k\) be the number of rounds between the \((k-1)\)-st and \(k\)-th useful recommendation.
By \Cref{prop:detectable-localization}, each \(X_k\) is stochastically dominated by a geometric random variable with
success probability \(\dprior/\sqrt{2}\).  Therefore
\[
    \E[X_k]\le \frac{\sqrt{2}}{\dprior},
    \qquad
    \E\left[\sum_{k=1}^N X_k\right]
    \le
    \frac{\sqrt{2}N}{\dprior}
    =
    O(\log\log T).
\]
Each physical round contributes at most \(\Delta U_{\max}\) regret. The hidden constant here is $\sqrt{2}\Delta U_{\max}/\dprior$.
\end{proof}



\subsection{Relevant Actions}
\label{app:critial-value}
\label{app:relevant_action}
Recall the relevant-action set
\[
    \actions_{\mathrm{rel}}
=
\Bigl\{a\in\actions:
\exists \post\in\simplex{\states}\text{ that }a\text{ uniquely maximizes }
\sum_{\omega\in\states}(\tbias\post(\omega)+(1-\tbias)\prior(\omega))u_R(a,\omega)
\Bigr\}.
\]
The following discussion shows that without excluding those actions, it is impossible to guarantee low regret.
\begin{remark}[Why tie-only actions are excluded]
Suppose \(a^\circ\notin\actions_{\mathrm{rel}}\).  Then the set of distorted beliefs at which
\(a^\circ\) can be selected is contained in a finite union of indifference hyperplanes.  Consequently,
for any non-degenerate interval \(J=[\underbias,\upbias]\), the set of Bayesian posteriors that induce
\(a^\circ\) for every \(\bias\in J\) has empty relative interior in \(\simplex{\states}\).  Thus any
scheme that relies on such an action is not stable to arbitrarily small perturbations of the unknown
bias parameter.

This is not only a technical inconvenience.  Consider a binary-state example in which \(a^\circ\)
is selected by tie-breaking only at the distorted belief \(\dpost(1)=1/2\), and suppose the sender
receives utility \(1\) from \(a^\circ\) and utility \(0\) from all other actions.  For each known
\(\bias\), a full-information sender can target the unique Bayesian posterior
\[
    \post_\bias(1)
    =
    \frac{1/2-(1-\bias)\prior(1)}{\bias}
\]
whenever it is feasible.  But if \(\bias\) is unknown, inducing \(a^\circ\) requires hitting this
posterior exactly.  Any interval-safe learning rule fails to induce \(a^\circ\) on a whole open interval
of possible bias values, while the full-information benchmark obtains a positive value.  Hence
including tie-only actions can create linear regret against the full-information benchmark.  
\end{remark}

We also recall the definition of interval-safe action region within $J$ and strict interval-safe region.
\[
    \areg{a}{J}=
    \bigcap_{\bias\in J}\areg{a}{\bias}
    =
    \Bigl\{
    \post\in\simplex{\states}:
    \Delta u_{a,a'}^\top\post\ge b_{a,a'}(J)=\max\cbr{b_{a,a'}(\underbias), b_{a,a'}(\upbias)},\ \forall a'\neq a
    \Bigr\}.
\]
\[
    \operatorname{int}_{\mathrm{IC}}\!\left(\areg{a}{J}\right)
    =
    \left\{
        \post\in\simplex{\states}:
        \Delta u_{a,a'}^\top\post>b_{a,a'}(J),
        \ \forall a'\neq a
    \right\}.
\]
We further show that for sufficiently small intervals $J\ni \tbias$, we can identify $A_{\mathrm{rel}}$ through checking whether $\operatorname{int}_{\mathrm{IC}}\!\left(\areg{a}{J}\right)\neq \emptyset$. After the localization stage returns an interval \(J=[\underbias,\upbias]\) of length \(O(1/\log T)\), for all sufficiently large \(T\), $J$ satisfies this length requirement.

\begin{lemma}[Local equivalence of relevant actions]
\label{lemma:relevant-action}
Fix the instance \(I\) and the true bias \(\tbias\). There exists
\(\epsilon_{\rm rel}(I,\tbias)>0\) such that for every interval
\(J\ni\tbias\) with \(|J|\le \epsilon_{\rm rel}(I,\tbias)\), and every action
\(a\in A\),
\[
a\in A_{\rm rel}
\quad\Longleftrightarrow\quad
\operatorname{int}_{\mathrm{IC}}\!\left(\areg{a}{J}\right)\neq \emptyset
\]
Consequently, on sufficiently small localized intervals, the computable
strict-feasibility check in \Cref{sec:interval-safe} keeps exactly the relevant actions.
\end{lemma}

\begin{proof}
Define the single-bias maximal slack
\[
s_a(\tbias)
:=
\max_{\nu\in\Delta(\Omega),\,s\in\mathbb R}
s
\quad
\text{s.t.}\quad
\Delta u_{a,a'}^\top\nu \ge b_{a,a'}(\tbias)+s,
\ \forall a'\neq a.
\]
By definition, \(a\in A_{\rm rel}\) if and only if \(s_a(\tbias)>0\).

First suppose \(\operatorname{int}_{\mathrm{IC}}\!\left(\areg{a}{J}\right)\neq \emptyset\). Then there exists \(\nu\in\Delta(\Omega)\) such
that
\[
\Delta u_{a,a'}^\top\nu>b_{a,a'}(J),
\qquad \forall a'\neq a.
\]
Since \(J\ni\tbias\), we have
\[
b_{a,a'}(J)\ge b_{a,a'}(\tbias)
\]
for every \(a'\neq a\). Hence the same posterior has strict IC slack at
\(\tbias\), so \(s_a(\tbias)>0\), and therefore \(a\in A_{\rm rel}\).

Conversely, suppose \(a\in A_{\rm rel}\). Then \(s_a(\tbias)>0\). Let
\[
\bar s_a:=s_a(\tbias)>0.
\]
Choose \(\nu_a\in\Delta(\Omega)\) such that
\[
\Delta u_{a,a'}^\top\nu_a
\ge
b_{a,a'}(\tbias)+\bar s_a/2,
\qquad \forall a'\neq a.
\]
The functions \(b_{a,a'}(\alpha)\) are Lipschitz on
\([\alpha_{\min}(I),1]\). Hence there exists \(\epsilon_a>0\) such that for
every interval \(J\ni\tbias\) with \(|J|\le\epsilon_a\),
\[
b_{a,a'}(J)
\le
b_{a,a'}(\tbias)+\bar s_a/4,
\qquad \forall a'\neq a.
\]
Therefore
\[
\Delta u_{a,a'}^\top\nu_a
\ge
b_{a,a'}(J)+\bar s_a/4,
\qquad \forall a'\neq a,
\]
so \(s_a(J)>0\) and $\operatorname{int}_{\mathrm{IC}}\!\left(\areg{a}{J}\right)\neq \emptyset$.

Taking the minimum of \(\epsilon_a\) over the finite set \(A_{\rm rel}\) gives
\(\epsilon_{\rm rel}(I,\tbias)>0\). This proves the equivalence.
\end{proof}

\subsection{Missing Details in Section \ref{sec:safe_exploration}}
\subsubsection{Full algorithm in the safe exploration stage}
\label{app:safeexplore-alg}
See the full safe exploration algorithm in \Cref{alg:sefi}, with subroutine in \Cref{alg:vertexsafescheme,alg:decompose}. We also introduce several LPs below, which are used in the algorithms.
\begin{algorithm}[ht]
\caption{\textsc{SafeExplore} on a fixed interval}
\label{alg:sefi}.
\begin{algorithmic}[1]
\Require Current interval \(J=[\underbias,\upbias]\), instance \(I\).
\State \(L\leftarrow |J|\), \(\eta\leftarrow L^2\).
\State Compute
\(
    \scheme_{\mathrm{vtx}}^J
    =
    \{(p_i,\vertex_i^J,a_i)\}_{i=1}^{M_J}
\)
using \Cref{alg:vertexsafescheme}. \Comment{\textsc{VertexSafeScheme}}.
\If{\(\up(\scheme_{\mathrm{vtx}}^J)=0\)}
    \State \Return \((J,\scheme_{\mathrm{vtx}}^J)\).
\EndIf
\For{each \(i\in\movv(\scheme_{\mathrm{vtx}}^J)\)}\Comment{\textsc{ChoosemovableBindingConstraint}}
    \State Choose one movable binding IC constraint
    \(
        \Delta u_{a_i,a_i'}^\top\post=b_{a_i,a_i'}(\beta_i),
        \qquad
        \beta_i\in\{\underbias,\upbias\}.
    \)
    \If{\(\beta_i=\underbias\)}
        \State Mark \(i\) as \textsc{Lower}.
    \Else
        \State Mark \(i\) as \textsc{Upper}.
    \EndIf
\EndFor
\State \((\ell,r)\leftarrow(\underbias,\upbias)\).
\State \(\scheme_{\mathrm{last}}\leftarrow\scheme_{\mathrm{vtx}}^J\).
\While{\(r-\ell>\eta\)}
    \For{each \(i\in\movv(\scheme_{\mathrm{vtx}}^J)\)}\Comment{\textsc{BuildProbe}}
        \If{\(i\) is \textsc{Lower}}
            \State \(m_i\leftarrow \ell+\eta\).
        \Else
            \State \(m_i\leftarrow r-\eta\).
        \EndIf
        \State Replace \(\vertex_i^J\) by the probe posterior
        \(\vertex_i^{\mathrm{pr}}(m_i)\) using Euclidean projection.
    \EndFor
    \State Leave non-informative posteriors unchanged.
    \State Add at most one correction posterior, if needed, to restore Bayes plausibility. \Comment{\Cref{lemma:utility_stability_changed}}
    \State Let the resulting probe scheme be \(\scheme_{\mathrm{probe}}^{J,\eta}(\ell,r)\).
    \State \(\scheme_{\mathrm{last}}\leftarrow \scheme_{\mathrm{probe}}^{J,\eta}(\ell,r)\).
    \State Play \(\scheme_{\mathrm{probe}}^{J,\eta}(\ell,r)\) until the first informative signal \(i\) is realized. \Comment{\textsc{RunUntilInformative}}
    \State Observe the receiver's action \(a_t\) on that signal.
    \If{\(i\) is \textsc{Lower}} \Comment{ \textsc{RefineInterval}}
        \If{\(a_t=a_i\)}
            \State \(\ell\leftarrow \ell+\eta\).
        \Else
            \State \Return \(([\ell,\ell+\eta],\scheme_{\mathrm{last}})\).
        \EndIf
    \Else
        \If{\(a_t=a_i\)}
            \State \(r\leftarrow r-\eta\).
        \Else
            \State \Return \(([r-\eta,r],\scheme_{\mathrm{last}})\).
        \EndIf
    \EndIf
\EndWhile
\State \Return \(([\ell,r],\scheme_{\mathrm{last}})\).
\end{algorithmic}
\end{algorithm}

\begin{algorithm}[ht]
\caption{\textsc{VertexSafeScheme} on interval \(J\)}
\label{alg:vertexsafescheme}
\begin{algorithmic}[1]
\Require Interval \(J=[\underbias,\upbias]\), instance \(\cI\).
\State  For every \(a\in A\), construct \(R_a^J\) and test whether
   \(\operatorname{int_{IC}}(R_a^J)\neq\emptyset\).
\State Solve the safe direct-recommendation LP (\eqref{eq:safe-lp}) on \(J\).
\State Obtain posterior/action pairs \(\{(\lambda_a,\bar\post_a,a):\lambda_a>0\}\).
\For{each \(a\) with \(\lambda_a>0\)}
    \State Decompose \(\bar\post_a\) into vertices of the polytope \(\areg{a}{J}\) using \Cref{alg:decompose}:
    \[
        \bar\post_a
        =
        \sum_{m=1}^{q_a}\gamma_{a,m}\vertex_{a,m}^J,
        \qquad
        \vertex_{a,m}^J\in\operatorname{vert}(\areg{a}{J}).
    \]
    \State Replace \((\lambda_a,\bar\post_a,a)\) by
    \(\{(\lambda_a\gamma_{a,m},\vertex_{a,m}^J,a)\}_{m=1}^{q_a}\).
\EndFor
\State Merge all resulting atoms and return
\[
    \scheme_{\mathrm{vtx}}^J
    =
    \{(p_i,\vertex_i^J,a_i)\}_{i=1}^{M_J}.
\]
\end{algorithmic}
\end{algorithm}

\begin{algorithm}[ht]
\caption{\textsc{VertexDecompose} a posterior in a safe polytope}
\label{alg:decompose}
\begin{algorithmic}[1]
\Require A polytope \(P=\areg{a}{J}\) and a posterior \(\post\in P\).
\State Initialize a working vertex set \(V\) by solving the initialization LP.
\Repeat
    \State Solve the restricted membership LP.
    \If{the membership LP is feasible}
        \State \Return the resulting decomposition.
    \EndIf
    \State Solve the separation LP to find \(c\) with
    \(c^\top\post>\max_{\vertex\in V}c^\top\vertex\).
    \State Solve the pricing LP to get a new vertex $\vertex^{\mathrm{new}}$, add it to \(V\).
\Until{membership is feasible}
\end{algorithmic}
\end{algorithm}

\paragraph{Safe direct-recommendation LP.}
A safe direct-recommendation scheme can be computed through the variables \(x_{a,\omega}\ge0\),
interpreted as the joint probability of state \(\omega\) and recommendation \(a\):
\begin{equation}
\begin{aligned}
\max_{x_{a,\omega}\ge 0}\quad
& \sum_{a\in A}\sum_{\omega\in\Omega} x_{a,\omega}u_S(a,\omega) \\
\text{s.t.}\quad
& \sum_{a\in A}x_{a,\omega}=\mu_0(\omega),
&& \forall \omega\in\Omega,\\
& \sum_{\omega\in\Omega}x_{a,\omega}\Delta u_{a,a'}(\omega)
\ge
b_{a,a'}(J)\sum_{\omega\in\Omega}x_{a,\omega},
&& \forall a\in A,\ \forall a'\in A\setminus\{a\},\\
& x_{a,\omega}=0,
&& \forall \omega\in\Omega,\ 
\forall a\in A\text{ such that }
\operatorname{intIC}(R_a^J)=\emptyset .
\end{aligned}
\label{eq:safe-lp}
\end{equation}
The last line is the strict-feasibility filtering step.

If \(\lambda_a:=\sum_{\omega}x_{a,\omega}>0\), the induced posterior is
\[
    \bar\post_a(\omega)
    :=
    \frac{x_{a,\omega}}{\lambda_a}
    \in \areg{a}{J}.
\]
Thus the LP returns a Bayes-plausible safe scheme
\[
    \{(\lambda_a,\bar\post_a,a):\lambda_a>0\}.
\]

\paragraph{Initialization LP.} Choose any generic vector $r\in\mathbb{Q}^{|\Omega|}$ and solve
\[
  v^{(1)}\in\arg\max_{x\in P} r^\top x,
\]
returning any optimal basic feasible solution; this is a vertex of $P$.

\paragraph{Restricted membership LP.} Given a current set $V=\{v^1,\dots,v^m\}\subseteq \mathrm{vert}(P)$, solve
\[
  \text{find }\lambda\in\mathbb{R}^m_+
  \quad\text{s.t.}\quad
  \mathbf{1}^\top\lambda=1,\qquad \sum_{j=1}^m \lambda_j v^j=\nu.
\]
If this LP is feasible, then $\nu\in\mathrm{conv}(V)$ and the resulting coefficients already give a decomposition.

\paragraph{Separation LP.} If the membership LP is infeasible, solve
\[
  \max_{c\in\mathbb{R}^{|\Omega|},\,\gamma\in\mathbb{R}}\ c^\top \nu-\gamma
  \quad\text{s.t.}\quad
  c^\top v^j\le \gamma\ \ (j=1,\dots,m),\qquad \|c\|_\infty\le 1.
\]
Any optimal solution with positive objective value yields a separating hyperplane $c$ such that
\[
  c^\top \nu>\max_{v\in V} c^\top v.
\]

\paragraph{Pricing LP.} Given such a $c$, solve
\[
  v^{\mathrm{new}}\in\arg\max_{x\in P} c^\top x,
\]
and again return any optimal basic feasible solution; this is an extreme point of $P$.

We then describe the omitted details in \Cref{sec:movingconstraint}.
\paragraph{Details of the local probe posterior construction in \Cref{sec:movingconstraint}}
Fix an informative vertex $\vertex \in R_a^J$ and a selected movable binding constraint $\Delta u_{a,a'}^\top v=b_{a,a'}(\beta)$, where
$\beta\in\{\underbias,\upbias\}$ is the endpoint determining
the interval-safe constraint. For a probe value $m\in J$, define
$v^{\mathrm{pr}}(m)$ as a minimum $l_2$ distance point satisfying the same simplex constraints, all non-selected IC constraints, and the modified equality
\[
    \Delta u_{a,a'}^\top v^{\mathrm{pr}}(m)=b_{a,a'}(m).
\]
\Cref{lem:rhs-stability} shows that this perturbation has distance
$O(|J|)$ from $v$. After perturbing all informative vertices, \Cref{lemma:utility_stability_changed}
adds at most one correction posterior to restore Bayes plausibility, losing only
$O(P_{\mathrm{info}}(\tau^J_{\mathrm{vtx}})|J|)$ utility.
\subsubsection{Geometric stability and perturbation bounds}
Before we show proofs of all propositions, we first list two useful lemmas here. For the definition of $\kappa_a$ and $\kappa$, see \Cref{app:general-constants}. These constants depend only on the constraint normals and not on the interval \(J\).

\begin{lemma}[Uniform distance bound under RHS perturbations]
\label{lem:rhs-stability}
Fix \(a\in\actions\).  For any two parameter values \(\beta_0,\beta_1\in[\alpha_{\min}(I),1]\)
and any \(\post_0\in\areg{a}{\beta_0}\),
\[
    \dist(\post_0,\areg{a}{\beta_1})
    \le
    \kappa_a
    \left\|
        \pos{\widetilde b_a(\beta_1)-\widetilde b_a(\beta_0)}
    \right\|_2
    \le
    \kappa_a
    \|\widetilde b_a(\beta_1)-\widetilde b_a(\beta_0)\|_2 .
\]
The same bound holds with \(\areg{J}{a}\) in place of \(\areg{a}{\beta_1}\) by taking
\(\widetilde b_a(J)\) as the right-hand side.
\end{lemma}

\begin{proof}
Write \(\post_0=\bar\post+Qz_0\).  Then \(z_0\in\widetilde R_a^{\beta_0}\), and
\[
    \dist(\post_0,\areg{a}{\beta_1})
    =
    \dist(z_0,\widetilde R_a^{\beta_1}).
\]
Define the componentwise violation vector
\[
    r
    :=
    \pos{\widetilde b_a(\beta_1)-\widetilde A_a z_0}
    \in\mathbb R^{m_a}_{+}.
\]
Since \(z_0\in\widetilde R_a^{\beta_0}\), we have
\(\widetilde A_a z_0\ge \widetilde b_a(\beta_0)\), and therefore
\[
    r
    \le
    \pos{\widetilde b_a(\beta_1)-\widetilde b_a(\beta_0)}
    \quad
    \text{componentwise}.
\]
Thus
\[
    \|r\|_2
    \le
    \left\|
        \pos{\widetilde b_a(\beta_1)-\widetilde b_a(\beta_0)}
    \right\|_2 .
\]

It remains to show
\(\dist(z_0,\widetilde R_a^{\beta_1})\le \kappa_a\|r\|_2\).  By definition,
\[
    \dist(z_0,\widetilde R_a^{\beta_1})
    =
    \min_y
    \{\|y\|_2:\widetilde A_a y\ge \widetilde b_a(\beta_1)-\widetilde A_a z_0\}.
\]
Since \(r\ge \widetilde b_a(\beta_1)-\widetilde A_a z_0\) componentwise, it suffices to consider
\[
    \min_y
    \{\|y\|_2:\widetilde A_a y\ge r\}.
\]
Let \(y^\star\) be a minimum-norm solution.  Let \(B\) be a full-row-rank subset of binding constraints
at \(y^\star\).  Then \(B\in\mathcal B_a\), and \(y^\star\) is the minimum-norm solution of
\[
    \widetilde A_{a,B}y=r_B.
\]
Hence
\[
    y^\star
    =
    \widetilde A_{a,B}^{\dagger}r_B,
    \qquad
    \|y^\star\|_2
    \le
    \|\widetilde A_{a,B}^{\dagger}\|_2\|r_B\|_2
    \le
    \kappa_a\|r\|_2 .
\]
Mapping back from \(z\)-coordinates to posteriors preserves Euclidean distance.  This proves the
first inequality.  The second follows from
\(\|\pos{x}\|_2\le\|x\|_2\).
\end{proof}

\begin{lemma}[Bayes-plausibility repair under small perturbations]
\label{lemma:utility_stability_changed}
Let
\(
    \scheme=\{(\lambda_i,\post_i,a_i)\}_{i=1}^k
\)
be Bayes-plausible:
\(
    \lambda_i\ge0,
    \sum_i\lambda_i=1,
    \sum_i\lambda_i\post_i=\prior .
\)
Suppose we perturb a subset of posteriors and obtain
\(\post_i'\in\simplex{\states}\) such that
\[
    \|\post_i'-\post_i\|_2\le \epsilon
    \quad\text{whenever }\post_i'\neq\post_i,
    \qquad
    \sum_{\post_i'\neq\post_i}\lambda_i\le p .
\]
Then there exists an augmented Bayes-plausible scheme
$
    \widehat\scheme
    =
    \{(\widehat\lambda_i,\post_i',a_i)\}_{i=1}^k
    \cup
    \{(\widehat\lambda_+,\post_+,\defaulta)\},
$
where \(\post_+\in \areg{J}{\defaulta}\), such that
\[
    |U_S(\scheme)-U_S(\widehat\scheme)|
    \le
    C_{\mathrm{rep}}\,p\epsilon ,\qquad
     C_{\mathrm{rep}}
    :=
    U_{\max}\left(\sqrt{|\states|}+\frac{4}{\dprior}\right).
\]
If no repair is needed, then \(\widehat\lambda_+=0\).
\end{lemma}

\begin{proof}
Let
\[
    \prior'
    :=
    \sum_i\lambda_i\post_i',
    \qquad
    r:=\prior-\prior' .
\]
Since both \(\prior\) and \(\prior'\) lie in the affine hull of the simplex,
\(\ones^\top r=0\).  Moreover,
\[
    \|r\|_2
    =
    \left\|
        \sum_i\lambda_i(\post_i-\post_i')
    \right\|_2
    \le
    \sum_{\post_i'\neq\post_i}
        \lambda_i\|\post_i-\post_i'\|_2
    \le
    p\epsilon .
\]
If \(r=0\), set \(\widehat\lambda_+=0\), \(\widehat\lambda_i=\lambda_i\), and we are done.

Assume \(r\neq0\).  Define
\[
    \widehat\lambda_+
    :=
    \frac{2\|r\|_2}{\dprior+2\|r\|_2},
    \qquad
    \widehat\lambda_i:=(1-\widehat\lambda_+)\lambda_i,
\]
and
\[
    \post_+
    :=
    \prior
    +
    \frac{1-\widehat\lambda_+}{\widehat\lambda_+}r .
\]
By construction,
\[
    \left\|
        \post_+-\prior
    \right\|_2
    =
    \frac{\dprior}{2}.
\]
Because \(\ones^\top r=0\), we have \(\ones^\top\post_+=1\).  By the standing margin convention,
\(\post_+\in\simplex{\states}\cap\areg{J}{\defaulta}\).

Bayes plausibility holds exactly:
\[
\begin{aligned}
    \sum_i\widehat\lambda_i\post_i'
    +
    \widehat\lambda_+\post_+
    &=
    (1-\widehat\lambda_+)\prior'
    +
    \widehat\lambda_+
    \left(
        \prior+
        \frac{1-\widehat\lambda_+}{\widehat\lambda_+}r
    \right)
\\
    &=
    (1-\widehat\lambda_+)\prior'
    +
    \widehat\lambda_+\prior
    +
    (1-\widehat\lambda_+)(\prior-\prior')
    =
    \prior .
\end{aligned}
\]
Also,
\[
    \widehat\lambda_+
    \le
    \frac{2\|r\|_2}{\dprior}
    \le
    \frac{2p\epsilon}{\dprior}.
\]

For utility, write
\[
    v_a(\post):=\sum_{\omega\in\states}\post(\omega)u_S(a,\omega).
\]
Then
\[
    |v_a(\post)-v_a(\post')|
    \le
    U_{\max}\sqrt{|\states|}\,\|\post-\post'\|_2.
\]
Therefore,
\[
\begin{aligned}
    |U_S(\scheme)-U_S(\widehat\scheme)|
    &\le
    U_{\max}\sqrt{|\states|}
    \sum_{\post_i'\neq\post_i}
        \lambda_i\|\post_i-\post_i'\|_2
    +
    U_{\max}\sum_i|\lambda_i-\widehat\lambda_i|
    +
    U_{\max}\widehat\lambda_+
\\
    &\le
    U_{\max}\sqrt{|\states|}\,p\epsilon
    +
    2U_{\max}\widehat\lambda_+
\\
    &\le
    U_{\max}\sqrt{|\states|}\,p\epsilon
    +
    \frac{4U_{\max}}{\dprior}p\epsilon .
\end{aligned}
\]
This proves the claim.
\end{proof}
\subsubsection{Proof of Proposition \ref{lemma:safe-optimal}}
Since $\scheme_{\mathrm{opt}}=\{(\lambda_i,\post_i,a_i)\}_{i=1}^k$
is a full-information optimal posterior-form scheme at the true bias \(\tbias\), where
$
    a_i=\bestresp{\post_i}{\tbias}.
$
Thus \(\post_i\in\areg{\tbias}{a_i}\).
We project each \(\post_i\) into the interval-safe region for the same action:
\[
    \post_i'
    \in
    \argmin_{\post\in\areg{J}{a_i}}
    \|\post-\post_i\|_2 .
\]
The localization stage ensures that \(\areg{J}{a_i}\neq\emptyset\) whenever
\(\areg{\tbias}{a_i}\neq\emptyset\), so this projection is well-defined.

For every \(a'\neq a_i\) and every \(\beta\in J\),
\[
    \left|
        b_{a_i,a'}(\beta)-b_{a_i,a'}(\tbias)
    \right|
    \le
    \sup_{\alpha\in[\alpha_{\min}(I),1]}
    \left|
        \frac{d}{d\alpha}b_{a_i,a'}(\alpha)
    \right|
    |J|.
\]
Since
\[
    \frac{d}{d\alpha}b_{a,a'}(\alpha)
    =
    \frac{C_{a,a'}}{\alpha^2},
\]
and \Cref{prop:detectable-localization} gives \(\alpha_{\min}(I)\ge\dprior/\sqrt{2}\), we have
\[
    \left|
        b_{a_i,a'}(\beta)-b_{a_i,a'}(\tbias)
    \right|
    \le
    \frac{2G_{\max}}{\dprior^2}|J| .
\]
Applying \Cref{lem:rhs-stability} yields an instance-dependent constant \(K_{\mathrm{dist}}=\frac{2G_{\max}\kappa}{\dprior^2}\) such
that
\[
    \|\post_i'-\post_i\|_2
    \le
    K_{\mathrm{dist}}|J|
    \qquad
    \forall i.
\]

Now apply \Cref{lemma:utility_stability_changed} to the perturbation
\(\post_i\mapsto\post_i'\), with \(p=1\) and
\(\epsilon=K_{\mathrm{dist}}|J|\).  We obtain a Bayes-plausible scheme
\(\widehat\scheme\) supported on \(R_{\mathrm{safe}}^J\) such that
\[
    U_S(\scheme_{\mathrm{opt}})-U_S(\widehat\scheme)
    \le
    C_{\mathrm{rep}}K_{\mathrm{dist}}|J| .
\]
Since \(\scheme_{\mathrm{vtx}}^J\) is optimal among interval-safe schemes,
\[
    U_S(\scheme_{\mathrm{vtx}}^J)
    \ge
    U_S(\widehat\scheme).
\]
Therefore
\[
    0
    \le
    U_S(\scheme_{\mathrm{opt}})-U_S(\scheme_{\mathrm{vtx}}^J)
    \le
    C_{\mathrm{rep}}K_{\mathrm{dist}}|J|
    =
    O(|J|).
\]

\subsubsection{Proof of Proposition \ref{lemma:safe-probe}}
Fix scan endpoints \((\ell,r)\subseteq J\).  Let
\[
    p_u:=\up(\scheme_{\mathrm{vtx}}^J).
\]
Only informative posteriors are perturbed, and their total probability is \(p_u\).  From \Cref{lem:rhs-stability}, for every
informative index \(i\),
\[
    \|\vertex_i^{\mathrm{pr}}(m_i)-\vertex_i^J\|_2
    \le
    K_{\mathrm{pr}}|J|,
\]
where $K_{\mathrm{pr}}=\frac{2G_{\max}\kappa}{\dprior^2}$, similar to the proof of \Cref{lemma:safe-optimal}.

Applying \Cref{lemma:utility_stability_changed} with
\[
    p=p_u,
    \qquad
    \epsilon=K_{\mathrm{pr}}|J|
\]
gives a Bayes-plausible probe scheme
\(\scheme_{\mathrm{probe}}^{J,\eta}(\ell,r)\) such that
\[
    \left|
        U_S(\scheme_{\mathrm{vtx}}^J)
        -
        U_S(\scheme_{\mathrm{probe}}^{J,\eta}(\ell,r))
    \right|
    \le
    C_{\mathrm{rep}}K_{\mathrm{pr}}\,p_u|J| .
\]
Thus
\[
    U_S(\scheme_{\mathrm{vtx}}^J)
    -
    U_S(\scheme_{\mathrm{probe}}^{J,\eta}(\ell,r))
    =
    O\!\left(\up(\scheme_{\mathrm{vtx}}^J)|J|\right).
\]

Moreover, the repair step scales the original atoms by a factor
\(1-O(p_u|J|)\) and adds only a default correction posterior.  Therefore, for all sufficiently small
intervals \(J\),
\[
    \up(\scheme_{\mathrm{probe}}^{J,\eta}(\ell,r))
    \ge
    \frac12\,\up(\scheme_{\mathrm{vtx}}^J).
\]

\subsubsection{Proof of Proposition \ref{prop:samplecomplexity}}
\paragraph{Finite vertex-index representation.}
For the proof of \Cref{prop:samplecomplexity}, we use a fixed finite index set for all possible vertices of the interval-safe regions.
Each interval-safe action region $\areg{a}{J}$ is a polytope, and its constraints consist of the simplex equality \(\ones^\top\post=1\), the IC
inequalities
$
    \Delta u_{a,a'}^\top\post \ge b^J_{a,a'}$,
and the non-negativity inequalities
$\post(\omega)\ge 0$.

To solve a possible candidate of vertex, we first fix one action $a$ and select \(n-1\) candidate binding
constraints chosen from the IC and non-negativity constraints of \(\areg{a}{J}\), then turn these constraints into equalities and add the simplex equality. This gives a linear system. If the linear system has a unque solution, it then corresponds to one possible vertex.
We use \emph{vertex index} 
$i=(B_i, a(i))$ to do such process for every possible action and set of constraints, 
where \(B_i\) is a collection of \(n-1\) candidate binding
constraints and \(a(i)\) is the corresponding recommending action. Then the linear system is denoted by
$
    \mA_i \post = \bb_i^J ,
$
where \(\mA_i\) depends only on the chosen constraint normals and not on $J$, while \(\bb_i^J\)
depends on \(J\) through the safe right-hand sides \(b^J_{a(i),a'}\).

We keep only
indices for which \(\mA_i\) is invertible, namely there exists a unique solution of the linear system. 
Use $\cI_{\mathrm{vtx}}^J$ to denote the set of such indexes.  
Note that firstly, since $\mA_i$ doesn't depend on $J$, then $\cI_{\mathrm{vtx}}^J$ doesn't depend on $J$, namely each index $i$ corresponds to the same action and the same set of constraints. By saying the same constraint, we mean it is the IC constraint on the same action pair or the same non-negative constraints.

Also, the actual set of vertices is actually a subset of $\cI_{\mathrm{vtx}}^J$. Besides, $\cI_{\mathrm{vtx}}^J$ may include infeasible points which violate other constraints, or duplicate points like degenerate cases or overlapping of different action regions. 
But since our goal is to have an uniform representation to denote the vertex-supported scheme, $\cI_{\mathrm{vtx}}^J$ is enough to use.
 Infeasible candidate vertices are assigned probability zero.

 For notational simplicity, enumerate
$\cI_{\mathrm{vtx}}^J=\cbr{1,\ldots,M}$. 
Each \(i\in\cI_{\mathrm{vtx}}^J\) therefore labels a possible vertex
\(\vertex_i^J\in\operatorname{vert}\rbr{\areg{a(i)}{J}}\). 

Thus any interval-safe vertex-supported scheme can be represented as
\[
    \scheme^J
    =
    \{(p_i,\vertex_i^J,a_i)\}_{i=1}^{M},
    \qquad
    p\in\simplex{\{1,\ldots,M\}},
    \qquad
    \sum_{i=1}^M p_i\vertex_i^J=\prior,
\]
with \(p_i=0\) whenever \(\vertex_i^J\) is infeasible.

\paragraph{Formal proof}
We prove the dichotomy by contradiction.  Suppose the first condition fails.  Then for every
\(n\ge1\), there exists an interval \(J_n=[\underbias_n,\upbias_n]\) with
\[
    \tbias\in J_n,
    \qquad
    |J_n|\le \frac1n,
\]
and an interval-safe vertex-supported optimizer
\[
    \scheme^n
    =
    \{(p_i^n,\vertex_i^{J_n},a_i)\}_{i=1}^{M}
\]
such that
\[
    \up(\scheme^n)\le \frac1n .
\]

Because the probability simplex in \(\mathbb R^M\) is compact, there is a subsequence, still denoted
by \(n\), such that
\[
    p^n\to p^\star .
\]
Since \(\underbias_n,\upbias_n\in[\alpha_{\min}(I),1]\), compactness also gives a subsequence such
that
\[
    \underbias_n\to\alpha^-,
    \qquad
    \upbias_n\to\alpha^+ .
\]
Because \(\tbias\in J_n\) and \(|J_n|\to0\), we have
\[
    \alpha^-=\alpha^+=\tbias.
\]
By continuity of each candidate vertex,
\[
    \vertex_i^{J_n}\to \vertex_i^{\tbias}
    \qquad
    \forall i.
\]
Define the limiting scheme
\[
    \scheme^\star
    :=
    \{(p_i^\star,\vertex_i^{\tbias},a_i)\}_{i=1}^{M}.
\]

We first show feasibility at \(\tbias\).  If \(p_i^\star>0\), then for all sufficiently large \(n\),
\(p_i^n>0\).  Since \(\scheme^n\) is feasible for \(J_n\), the candidate posterior
\(\vertex_i^{J_n}\) satisfies all IC and nonnegativity constraints of action \(a_i\) under \(J_n\).
Taking limits gives
\[
    \Delta u_{a_i,a'}^\top \vertex_i^{\tbias}
    \ge
    b_{a_i,a'}(\tbias),
    \qquad
    \forall a'\neq a_i,
\]
and also \(\vertex_i^{\tbias}\in\simplex{\states}\).  Bayes plausibility follows from
\[
    \sum_i p_i^n\vertex_i^{J_n}=\prior
\]
by taking limits.  Therefore \(\scheme^\star\) is feasible at the true bias \(\tbias\).

Next we show optimality.  Since \(\scheme^n\) is safe-optimal on \(J_n\),
\[
    U_S(\scheme^n)=\OPT^{\mathrm{safe}}(J_n).
\]
Because \(J_n\) tightens the IC constraints relative to the true-bias problem,
\[
    \OPT^{\mathrm{safe}}(J_n)\le \iopt{I}{\tbias}.
\]
On the other hand, \Cref{lemma:safe-optimal} gives
\[
    \OPT^{\mathrm{safe}}(J_n)
    \ge
    \iopt{I}{\tbias}-C_{\mathrm{rep}}K_{\mathrm{dist}}|J_n|
\]
where constants are only related to instance $\cI$, but not $J$.
Hence
\[
    \OPT^{\mathrm{safe}}(J_n)\to \iopt{I}{\tbias}.
\]
By continuity of the finite representation,
\[
    U_S(\scheme^n)\to U_S(\scheme^\star).
\]
Therefore
\[
    U_S(\scheme^\star)=\iopt{I}{\tbias},
\]
so \(\scheme^\star\) is optimal at the true bias.

Finally, since \(\up(\scheme^n)\le1/n\), we must have
\[
    \up(\scheme^\star)=0.
\]
Indeed, if some movable informative index had \(p_i^\star>0\), then \(p_i^n\ge p_i^\star/2\) for
all large \(n\), contradicting \(\up(\scheme^n)\le1/n\).  Thus the second condition holds: there
exists a true-bias optimal scheme with zero informative probability.

The remaining statement we will show is that if such a
zero-informative optimal scheme exists, then there exists \(\delta_{\mathrm{stab}}>0\), for every interval \(J\ni\tbias\) with \(|J|\le\delta_{\mathrm{stab}}\), the same scheme
\(\scheme^\star\) is feasible and optimal for the interval-safe problem on \(J\).  Consequently,
$
    \OPT^{\mathrm{safe}}(J)=\iopt{I}{\tbias}.$

Write the finite-support optimal scheme as
\[
    \scheme^\star=\{(p_i,\post_i,a_i)\}_{i=1}^k,
    \qquad
    p_i>0,
    \qquad
    \sum_i p_i\post_i=\prior,
\]
where \(a_i=\bestresp{\post_i}{\tbias}\).

For every action pair \((a,a')\), define the IC slack
\[
    \operatorname{slack}_{a,a'}(\post,\bias)
    :=
    \Delta u_{a,a'}^\top\post-b_{a,a'}(\bias).
\]
A movable constraint is one with \(C_{a,a'}\neq0\).  The condition
\(\up(\scheme^\star)=0\) means that no support posterior of \(\scheme^\star\) lies on a movable
binding IC constraint.  Hence, for every support atom \(i\) and every movable pair
\((a_i,a')\),
\[
    \operatorname{slack}_{a_i,a'}(\post_i,\tbias)>0.
\]
Because the support is finite, the minimum positive movable slack
\[
    s^\star
    :=
    \min_{i}
    \min_{\substack{a'\neq a_i\\ C_{a_i,a'}\neq0}}
    \operatorname{slack}_{a_i,a'}(\post_i,\tbias)
\]
is strictly positive, ignoring pairs over an empty set.

By uniform continuity of \(b_{a,a'}(\cdot)\) on
\([\alpha_{\min}(I),1]\), there exists \(\delta_{\mathrm{stab}}>0\) such that whenever
\(J\ni\tbias\) and \(|J|\le\delta_{\mathrm{stab}}\),
\[
    b^J_{a,a'}
    \le
    b_{a,a'}(\tbias)+\frac{s^\star}{2}
\]
for every movable pair that appears in the support of \(\scheme^\star\).  Therefore all movable IC
constraints remain satisfied on \(J\).  Non-movable constraints have \(C_{a,a'}=0\), so
\(b_{a,a'}(\bias)\equiv0\) and are unchanged.  Thus every support posterior of \(\scheme^\star\)
lies in its corresponding interval-safe region, and \(\scheme^\star\) is feasible for the safe problem on
\(J\).

Since the safe feasible set is a subset of the feasible set at the true bias \(\tbias\),
\[
    \OPT^{\mathrm{safe}}(J)\le \iopt{I}{\tbias}.
\]
Feasibility of \(\scheme^\star\) gives the reverse inequality.  Hence
\[
    \OPT^{\mathrm{safe}}(J)=\iopt{I}{\tbias},
\]
and \(\scheme^\star\) is safe-optimal on \(J\).

\subsubsection{Proof of Proposition \ref{prop:safe_exploration_regret}}
Consider a safe-exploration phase starting from interval
\(J_r=[\underbias_r,\upbias_r]\) of length \(L_r\).  The algorithm uses step size
\[
    \eta_r:=L_r^2 .
\]
Each safe informative realization moves either the left endpoint or the right endpoint inward by
\(\eta_r\).  A unsafe informative realization localizes \(\tbias\) to an adjacent interval of length
\(\eta_r\).  Therefore the phase contains at most
\[
    \left\lceil \frac{L_r}{\eta_r}\right\rceil
    =
    O\!\left(\frac1{L_r}\right)
\]
safe probes and at most one unsafe probe.  In all cases,
\[
    L_{r+1}\le \eta_r=L_r^2 .
\]

We now bound the regret of one phase.  Let
\[
    p_u:=\up(\scheme_{\mathrm{vtx}}^{J_r})
\]
for the safe vertex-supported optimizer used at the start of the phase.  By the proof of \Cref{lemma:safe-probe}, the
corresponding probe scheme has informative-signal probability at least \(p_u/2\) on sufficiently small
intervals.

There are two cases from \Cref{prop:samplecomplexity}.

\paragraph{Case 1: informative probability is bounded below.}
There exist constants \(c>0\) and \(\epsilon>0\) such that for every interval
\(J\ni\tbias\) with \(|J|\le\epsilon\), every interval-safe vertex-supported optimizer satisfies
\[
    \up(\scheme_{\mathrm{vtx}}^J)\ge c .
\]
Then every probe has informative-signal probability at least \(c/2\), so the expected waiting time for
an informative realization is \(O(1/c)\).

By \Cref{lemma:safe-optimal}, the safe scheme is \(O(L_r)\)-optimal relative to the full-information benchmark.
By \Cref{lemma:safe-probe}, the probe perturbation changes utility by at most
\(O(p_u|L_r|)\le O(|L_r|)\).    Multiplying by the expected waiting time, one safe probe contributes
\(O(|L_r|/c)\) expected regret.  Since there are \(O(|L_r|/|L_r|^2)=O(1/|L_r|)\) safe probes, their total contribution in
this phase is \(O(1/c)\).

There is at most one unsafe probe.  Its per-round regret is at most \(\Delta U_{\max}\), and its
expected waiting time is \(O(1/c)\), so it contributes \(O(1/c)\).  Thus the whole phase contributes
\(O(1)\) expected regret, with constants depending only on the instance.

\paragraph{Case 2: a zero-informative optimal scheme exists.}
For all sufficiently small intervals \(J_r\ni\tbias\),
\[
    \OPT^{\mathrm{safe}}(J_r)=\iopt{I}{\tbias}.
\]
Hence the baseline safe scheme already achieves the full-information benchmark.  The only loss comes
from the probe perturbation.

For an safe probe, \Cref{lemma:safe-probe} gives per-round regret \(O(p_u|L_r|)\), while the expected
waiting time is \(O(1/p_u)\).  Thus one safe probe contributes \(O(|L_r|)\).  Summing over
$O(|L_r|/|L_r|^2)=O(1/|L_r|)$ safe probes gives
\[
    O\!\left(\frac{1}{|L_r|}\cdot |L_r|\right)
    =
    O(1).
\]
For the single unsafe probe, only informative posteriors can induce the wrong action.  Their total
probability is \(O(p_u)\), and the expected waiting time is \(O(1/p_u)\), so the unsafe probe
contributes \(O(1)\).  Therefore this phase also has \(O(1)\) expected regret.

\paragraph{Number of phases.}
The interval lengths satisfy
\[
    L_{r+1}\le L_r^2,
    \qquad
    L_0\le \frac1{\log T}.
\]
Thus
\[
    L_r
    \le
    \left(\frac1{\log T}\right)^{2^r}.
\]
The smallest \(r\) such that \(L_r\le 1/T\) is \(O(\log\log T)\).  Therefore the total regret over all
safe-exploration phases is \(O(\log\log T)\).

\subsection{Proof of Theorem \ref{thm:generalbound}}
\subsubsection{Runtime of the general algorithm}
The localization stage solves the threshold-test LP, which has
\(|\actions||\states|\) variables and polynomially many constraints.  The safe direct-recommendation
LP used in \textsc{VertexSafeScheme} has the same polynomial scale.

For vertex decomposition, the polytope \(\areg{J}{a}\) is given explicitly by polynomially many
linear inequalities.  Each call to \textsc{VertexDecompose} uses only polynomial-size LPs: a generic
linear optimization LP over \(\areg{J}{a}\), a restricted membership LP over the currently discovered
vertices, a separation LP, and a pricing LP over \(\areg{J}{a}\).  Equivalently, one may use the
standard separation--optimization implementation of convex decomposition to obtain a formal
polynomial-time routine.  The final Carathéodory compression is a linear-algebraic elimination step.

The probe construction only solves constant-size linear systems associated with active vertex bases
and then applies the Bayes-plausibility repair step.  Therefore each update of the online algorithm is
polynomial in the finite instance description, and the total runtime up to horizon \(T\) is polynomial
in \(T\), \(|\actions|\), and \(|\states|\).

\subsubsection{Regret of the general algorithm}
We decompose the regret into the localization stage, the
safe-exploration stage, and the final commitment stage.

\paragraph{Localization stage.}
The localization stage starts from the interval \([\alpha_{\min}(I),1]\) and performs binary search
using \textsc{ThresholdTest}.  It stops once the interval length is at most \(1/\log T\).  Therefore it
requires \(O(\log\log T)\) informative threshold-test realizations.  By \Cref{cor:localization-regret},
the expected localization regret is \(O(\log\log T)\).


After this step, no action region changes feasibility status inside the current interval.  Let this
interval be \(J_0\), with length \(L_0\le 1/\log T\).

\paragraph{Safe exploration stage.}
From \Cref{prop:safe_exploration_regret}, the expected total regret in this stage is $O(\log\log T)$.

\paragraph{Commitment stage.}
If safe exploration reaches an interval \(J_R\) with \(|J_R|\le1/T\), the algorithm commits to an
interval-safe vertex-supported optimizer on \(J_R\).  \Cref{lemma:safe-optimal} implies that its per-round gap
relative to the full-information benchmark is \(O(1/T)\), so the total commitment regret is \(O(1)\).
If instead \textsc{SafeExplore} stops early because
\(\up(\scheme_{\mathrm{vtx}}^{J_R})=0\), then Case 1 of \Cref{prop:samplecomplexity} cannot hold on sufficiently
small intervals.  Hence Case 2 holds, and this implies that the current safe scheme is
already optimal at the true bias.  Committing to it incurs no additional asymptotic regret.

Combining the localization stage, the
\(O(\log\log T)\) safe-exploration phases, and the final commitment stage gives
\[
    \ireg{T}{\cI}{\Pi^{\mathrm{GSE}}}{\tbias}
    =
    O(\log\log T).
\]

\section{Receiver Welfare under Biased Persuasion}\label{sec_app:receiver-extension}
This appendix compares the receiver's ex ante utility under persuasion with her no-persuasion utility.
All notation follows the model section. 
In the classical setting \cite{kamenica2011bayesian}, the receiver's utility cannot decrease under persuasion: additional information cannot reduce the value of an optimal decision, because the receiver can always ignore the signal and play the default action.

In our setting, the receiver is non-Bayesian. Even if the sender provides more information, it is not clear whether the receiver's true expected utility increases, because the receiver chooses actions based on a biased belief. We show that under linear bias model, the non-decreasing guarantee still holds.

\begin{proposition}[Preservation of consistency]
\label{prop:linear_consistency}
If $\postdist$ is Bayes-plausible, then the induced distribution of distorted posteriors also has mean $\prior$:
$
\E_{\post \sim \postdist}\big[\tbias\post+(1-\tbias)\prior\big] = \prior.
$
\end{proposition}
\begin{proposition}[Persuasion does not hurt the biased receiver]
\label{prop:receiver_not_hurt_linear_bias}
For any bias level $\tbias$ and any feasible signal scheme
$\tau=\{(p_i,\nu_i)\}_{i=1}^k$, the receiver's expected utility under $\tau$ is at least her no-persuasion utility.
\end{proposition}
\begin{proof}
Consider any feasible signal scheme $\tau=\{p_i,\post_i\}_{i=1}^k$. Under the signal scheme, the receiver's utility can be rewritten as below.
    \begin{align*}
        \sum_{i} p_i\sum_{\omega}u_R(a^*(\nu_i;\tbias),\omega)\nu_i(\omega)&=\sum_i p_i\sum_{\omega}u_R(a^*(\nu_i^{\tbias}),\omega)\left(\frac{1}{\tbias}\nu_i^{\tbias}(\omega)-\frac{1-\tbias}{\tbias}\prior(\omega)\right)\tag{$\nu_i^{\tbias}=\tbias\nu_i+(1-\tbias)\prior$}\\
        &\geq\sum_i p_i\sum_{\omega}u_R(a_0,\omega)\frac{1}{\tbias}\nu_i^{\tbias}(\omega)-\frac{1-\tbias}{\tbias}\sum_{\omega}u_R(a_0,\omega)\prior(\omega)\tag{ $a^*(\nu_i^{\tbias})$ is the optimal action facing belief $\nu_i^{\tbias}$, $a_0$ is the optimal action facing belief $\prior$}\\
        &=\sum_\omega u_R(a_0,\omega)\prior(\omega)\tag{$\sum_i p_i\nu_i^{\tbias}=\prior$ from \Cref{prop:linear_consistency}, $\frac{1}{\tbias}-\frac{1-\tbias}{\tbias}=1$}
    \end{align*}

Since $\sum_\omega u_R(a_0,\omega)\prior(\omega)$ is the receiver's utility without persuasion, we conclude the proof.
\end{proof}

Another question is: for a fixed bias level $\tbias$, how does the receiver's utility compare between
(i) a sender who correctly accounts for the receiver's bias when optimizing the signal scheme, and
(ii) a sender who incorrectly treats the receiver as Bayesian when optimizing the signal scheme?
The comparison can go in either direction.

The main intuition is that when the sender does not know the receiver’s true bias level, the sender cannot precisely target the receiver’s decision thresholds. This imprecision can cut in either direction: the sender may fail to extract as much utility as under correct knowledge, which can leave the receiver better off, or the sender may inadvertently push the receiver into suboptimal actions under the true bias, reducing the receiver’s utility and potentially lowering the sender’s utility as well.

\begin{example}
\label{ex:receiver-nonmonotone}
Let \(\states=\{0,1\}\), \(\actions=\{a_1,a_2,a_3\}\), and \(\prior(1)=0.1\). The receiver's utility are
\[
\begin{aligned}
&u_R(a_1,0)=u_R(a_1,1)=0,\qquad
u_R(a_2,0)=-1,\quad u_R(a_2,1)=4,\\
&u_R(a_3,0)=-51,
\qquad u_R(a_3,1)=54.
\end{aligned}
\]
For a binary belief \(q=\Pr(\omega=1)\), the receiver obtains
\[
    \E[u_R(a_2,\omega)\mid q] = -1+5q,
    \qquad
    \E[u_R(a_3,\omega)\mid q] = -51+105q.
\]
Thus she chooses \(a_1\) for \(q\le 0.2\), \(a_2\) for \(0.2\le q\le 0.5\), and \(a_3\) for \(q\ge 0.5\). Let the sender's utility depend only on the induced action:
\[
    u_S(a_1)=0,
    \qquad u_S(a_2)=2,
    \qquad u_S(a_3)=10 .
\]

First take \(\tbias=0.2\). A sender who mistakenly treats the receiver as Bayesian chooses the two-posterior scheme \(q_L=0\), \(q_H=0.5\), with probabilities \(0.8\) and \(0.2\). Under the true bias,
\[
    \hat q_L=0.2\cdot 0+0.8\cdot 0.1=0.08,
    \qquad
    \hat q_H=0.2\cdot 0.5+0.8\cdot 0.1=0.18,
\]
so the receiver chooses \(a_1\) after both signals and obtains utility \(0\). By contrast, a sender who accounts for \(\tbias=0.2\) cannot induce \(a_3\), and induces \(a_2\) using \(q_L=0\), \(q_H=0.6\), with probabilities \(5/6\) and \(1/6\). The receiver's utility is
\[
    \frac{1}{6}(-1+5\cdot 0.6)=\frac{1}{3}.
\]
Hence the receiver is better off when the sender knows \(\tbias\).

Now take \(\tbias=1/3\). The same Bayesian-designed scheme \((q_L,q_H)=(0,0.5)\), with probabilities \((0.8,0.2)\), induces distorted high belief
\[
    \hat q_H=\frac{1}{3}\cdot 0.5+\frac{2}{3}\cdot 0.1=\frac{7}{30},
\]
so the receiver chooses \(a_2\) after the high signal and obtains
\[
    0.2(-1+5\cdot 0.5)=0.3.
\]
A sender who correctly accounts for \(\tbias=1/3\) again cannot induce \(a_3\). To induce \(a_2\), the minimal Bayesian posterior is \(q_H=0.4\), so the optimal scheme uses \(q_L=0\), \(q_H=0.4\), with probabilities \(0.75\) and \(0.25\). The receiver obtains
\[
    0.25(-1+5\cdot 0.4)=0.25.
\]
Hence the receiver is worse off when the sender knows \(\tbias\).
\end{example}

\section{Extension: Changing Priors and Utility Functions} \label{sec_app:changing_prior}
We now allow the prior and utility functions to vary across rounds while the receiver's bias \(\tbias\) remains fixed. The state and action spaces \(\states\) and \(\actions\) are fixed and finite. At round \(t\), a context
\[
    x_t=(\mu_{x_t},u_{S,x_t},u_{R,x_t})
\]
is realized and observed by the sender before she chooses a signaling scheme. Conditional on \(x_t=x\), the state is drawn from \(\mu_x\). 
We do \emph{not} require the sender to know the full family of possible contexts
in advance, nor do we impose any stochastic assumption on the sequence $(x_t)_{t=1}^T$.
Instead, we only assume that there exists a context set $\mathcal X$ such that $x_t \in \mathcal X$
for every $t$.

For each context \(x\), write \(I_x=(\states,\actions,\mu_x,u_{S,x},u_{R,x})\). Define
\[
\begin{aligned}
    a_x^\star(\post;\bias)
    &\in \argmax_{a\in\actions}
    \sum_{\omega\in\states}
    \big((1-\bias)\mu_x(\omega)+\bias\post(\omega)\big)u_{R,x}(a,\omega),\\
    R_{x,a}^{\bias}
    &:=\{\post\in\simplex{\states}:a_x^\star(\post;\bias)=a\},\\
    V_x^{\bias}(\post)
    &:=\sum_{\omega\in\states}\post(\omega)u_{S,x}\big(a_x^\star(\post;\bias),\omega\big),\\
    \OPT_x(\bias)
    &:=\sup_{\postdist\in\Delta(\simplex{\states}):\E_{\post\sim\postdist}[\post]=\mu_x}
        \E_{\post\sim\postdist}\big[V_x^{\bias}(\post)\big].
\end{aligned}
\]
The dynamic oracle knows \(\tbias\) and, in every round \(t\), plays an optimal Bayes-plausible scheme for the realized context \(x_t\). For a fixed context sequence \(\mathbf{x}=(x_1,\ldots,x_T)\), define the expected dynamic regret by
\[
    \Reg_T^{\mathbf{x}}(\Pi;\tbias)
    :=\sum_{t=1}^T
    \left(
    \OPT_{x_t}(\tbias)
    -\E_{\Pi}\big[V_{x_t}^{\tbias}(\post_t)\big]
    \right).
\]

\begin{assumption}[Uniform contextual regularity]
\label{ass:contextual-regularity}
For every \(x\in\mathcal X\), the instance \(I_x\) satisfies the \Cref{ass:strict_default}, together with the tie-breaking rule of the relevant-action restriction. In addition:
\begin{enumerate}
    \item The true bias is detectable in every context: \(\tbias\ge \bar\alpha_{\min}:=\max_{x\in\mathcal X}\alpha_{\min}(I_x)\).
    \item The sender utilities are uniformly bounded: \(U_{\max}^{\mathcal X}:=\max_{x\in\mathcal X}\max_{a,\omega}|u_{S,x}(a,\omega)|<\infty\).
      \item For every context $x$, the analogue of \Cref{prop:samplecomplexity} holds, and there exists a positive constant $c= \min_{x} c_x >0,\epsilon=\min_x \epsilon_x>0$ for the first branch in \Cref{prop:samplecomplexity}, and a positive constant $\delta=\min_x \delta_x>0$ for the second branch.
    \item The constants in \Cref{thm:generalbound} can be chosen uniformly over \(x\in\mathcal X\), see more details in \Cref{app:general-constants}. Equivalently, there exist positive constant $\dprior=\min_{x\in\mathcal X}\delta_{\mu_0,x}$ and 
finite constants $\kappa,G_{\max},L_b$ satisfying $$\kappa=\max_{x\in\mathcal X}\kappa_x,\quad G_{\max}=\max_{x\in\mathcal X}G_{\max,x},\quad L_b=\max_{x\in\mathcal X}L_{b,x}.$$
\end{enumerate}
\end{assumption}
Note that the above assumption naturally holds if the context set $\mathcal X$ is finite.

\paragraph{Contextual algorithm.}
The contextual algorithm maintains one global confidence interval \(J\) for the common parameter \(\tbias\). In each round, every LP, threshold test, safe scheme, and probe is instantiated with the realized context \(x_t\). Thus the algorithm is the contextual analogue of \Cref{algo:G-SETC}. 





\begin{theorem}
\label{thm:contextual-gse}
Under \Cref{ass:contextual-regularity}, for every fixed context sequence \(\mathbf{x}\in\mathcal X^T\), the contextual version of \Cref{algo:G-SETC} \(\Pi_{\mathrm{CGSE}}\) satisfies that 
\[
    \Reg_T^{\mathbf{x}}(\Pi_{\mathrm{CGSE}};\tbias)=O(\log\log T).
\]
\end{theorem}
\begin{proof}
The following proof is a contextual analogue of proof of \Cref{thm:generalbound}.

\paragraph{Localization.}
The localization stage runs threshold tests at midpoints of the current global interval until its length is at most \(1/\log T\). Since the interval is contained in \([\bar\alpha_{\min},1]\), the threshold test is feasible in every realized context. By the uniform version of \Cref{prop:detectable-localization}, in every round of a threshold test the conditional probability of an informative non-default recommendation is at least \(p_{\min}=\dprior/\sqrt{2}>0\). Therefore the waiting time for one informative comparison is stochastically dominated by a geometric random variable with mean \(1/p_{\min}\), even though contexts may vary over time. The number of informative comparisons needed to reduce the interval from constant length to \(1/\log T\) is \(O(\log\log T)\). Since one-round regret is bounded by \(2U_{\max}^{\mathcal X}\), localization contributes \(O(\log\log T)\) expected regret.


\paragraph{One safe-exploration phase.}
Consider a safe-exploration phase that starts with interval \(J\) and step size \(\eta=L^2\). For each realized context \(x\), let \(\tau_{x,\mathrm{vtx}}^J\) be the contextual interval-safe vertex-supported optimizer, and let
\[
    p_x^J=P_{\mathrm{info}}(\tau_{x,\mathrm{vtx}}^J).
\]
If \(p_x^J=0\), the algorithm plays \(\tau_{x,\mathrm{vtx}}^J\) and does not attempt to update the interval in that round. If \(p_x^J>0\), it plays the contextual probe constructed using the binding-movable-constraint construction from \Cref{sec:movingconstraint}. By the uniform version of \Cref{lemma:safe-probe}, the probability of an informative realization under this probe is at least \(p_x^J/2\) for all sufficiently small intervals.

We first bound the regret accumulated until the next informative realization that corresponds to a safe probe, i.e., a realization that moves one endpoint inward by \(\eta\). There are two cases, matching the dichotomy in Proposition~8. If context \(x\) is in the first branch of the dichotomy, then \(p_x^J\ge c\), so the informative probability in that round is at least \(c/2\). The safe-optimal gap is at most $O(L)$, and the probe perturbation adds at most $O(L)$. Thus the one-round regret on such a safe probe is at most $O(L)$, while its informative probability is at least \(c/2\).

If context \(x\) is in the second branch, then the contextual safe optimum already attains \(\OPT_x(\tbias)\) on sufficiently small intervals. The only loss is the probe perturbation, at most \(O(p_x^JL)\), while the informative probability is at least \(p_x^J/2\). Hence the one-round regret is at most $O(L)$.

Combining the two cases, there is a uniform constant \(C\) such that, on every safe-probe round before the next safe informative realization,
\[
    \text{one-round regret}\le C L\cdot
    \Pr(\text{informative realization in that round}\mid\mathcal F_{t-1}).
\]
Let \(\sigma\) be the first informative realization time for the current probe. The standard hazard-rate identity gives
\[
    \E\left[\sum_{t\le\sigma}
    \Pr(\sigma=t\mid\mathcal F_{t-1})\right]
    \le 1.
\]
Therefore the expected regret before one safe endpoint update is at most \(CL\). During a phase, each safe informative realization moves one endpoint by \(\eta=L^2\), so there are at most \(\lceil L/\eta\rceil=O(1/L)\) such updates. The total expected regret from safe updates in the phase is therefore \(O(1)\).

There can be at most one unsafe informative realization in the phase, because once such a realization occurs the algorithm localizes \(\tbias\) to an adjacent interval of length \(\eta\) and ends the phase. In the first branch of the dichotomy, the informative probability is uniformly bounded below, and one-round regret is uniformly bounded by \(2U_{\max}^{\mathcal X}\), so the unsafe realization contributes \(O(1)\) expected regret. In the second branch, the baseline safe scheme is optimal and only the perturbed informative mass can induce the wrong action; its total probability is \(O(p_x^J)\), while the informative probability is at least \(p_x^J/2\). Thus the unsafe realization also contributes \(O(1)\) expected regret. Hence every completed safe-exploration phase contributes \(O(1)\) expected regret. A phase truncated by the horizon is bounded by the same argument.

\paragraph{Number of phases and commitment.}
At the end of any completed safe-exploration phase, the new interval length satisfies
\[
    L_{r+1}\le \eta_r=L_r^2.
\]
After localization, \(L_0\le 1/\log T\). Therefore after \(O(\log\log T)\) phases the interval length is at most \(1/T\). When this happens, the algorithm commits, in each realized context \(x_t\), to the contextual interval-safe optimizer for the final interval. The uniform safe-optimal gap gives per-round regret $O(1/T)$, and hence the total commitment regret is \(O(1)\).

Adding the localization regret, the constant diagnostic cost, the \(O(\log\log T)\) safe-exploration phases, and the final commitment cost yields the claimed bound.

\end{proof}

\section{Extension: Jointly Unknown Prior and Bias}\label{sec_app:prior-bias-extension}
\subsection{Binary case with jointly unknown prior and bias.}
We now consider the binary case in which both the receiver's bias and the prior are fixed across rounds and unknown to the sender. As in Section~\ref{sec:binary}, let $\Omega=\{0,1\}$, $A=\{0,1\}$, $\mu^\ast:=\Pr(\omega=1)\in(0,1)$, and $\alpha^\ast\in[0,1]$. The sender utility is $u_S(a,\omega)=\mathbf{1}\{a=1\}$, so the sender only cares about inducing action $1$. For a realized signal $s$, let $\nu=\Pr(\omega=1\mid s)$ be the Bayesian posterior. Under the linear-bias model, the receiver evaluates this signal using the distorted posterior $(1-\alpha^\ast)\mu^\ast+\alpha^\ast\nu$.

Let $\hat q\in(0,1)$ be the receiver's cutoff in distorted-belief space: the receiver takes action $1$ iff $(1-\alpha^\ast)\mu^\ast+\alpha^\ast\nu\ge \hat q$. Equivalently, there is a Bayesian-posterior threshold $\nu^\ast$ such that the receiver takes action $1$ iff $\nu\ge \nu^\ast$, where $\nu^\ast$ is determined by $(1-\alpha^\ast)\mu^\ast+\alpha^\ast\nu^\ast=\hat q$, i.e.,
\begin{align*}
\nu^\ast=\frac{\hat q-(1-\alpha^\ast)\mu^\ast}{\alpha^\ast}
=\mu^\ast+\frac{\hat q-\mu^\ast}{\alpha^\ast}.
\end{align*}
Thus $\nu^\ast$ is the critical threshold in posterior space.

Compared with the known-prior case, the difficulty is that when $\mu^\ast$ is known, one can guess $\alpha$, compute the target posterior $\nu(\alpha)=\frac{\hat q-(1-\alpha)\mu^\ast}{\alpha}$, and implement the corresponding binary scheme supported on $\{0,\nu(\alpha)\}$. When $\mu^\ast$ is unknown, this direct reduction is unavailable because implementing a target posterior itself depends on the prior. Nevertheless, under a natural two-signal family, the sender still faces a one-dimensional threshold-learning problem, much as in the binary-action unknown-prior model of \cite{Li_Lin_2025}.

Specifically, consider the implementable family $\{\pi_m\}_{m\in[0,1]}$ given by $\pi_m(\mathrm{High}\mid \omega=1)=1$ and $\pi_m(\mathrm{High}\mid \omega=0)=m$. Thus the sender always sends $\mathrm{High}$ in state $1$ and sends $\mathrm{High}$ with probability $m$ in state $0$. Under $\pi_m$, Bayes' rule gives
$\nu(m;\mu^\ast)=\Pr(\omega=1\mid \mathrm{High})
=\frac{\mu^\ast}{\mu^\ast+(1-\mu^\ast)m}$.
Hence the receiver accepts $\mathrm{High}$ iff $(1-\alpha^\ast)\mu^\ast+\alpha^\ast\frac{\mu^\ast}{\mu^\ast+(1-\mu^\ast)m}\ge \hat q$. This motivates the critical implementable threshold $m^\ast$, defined as the unique solution of
\begin{align*}
(1-\alpha^\ast)\mu^\ast+\alpha^\ast\frac{\mu^\ast}{\mu^\ast+(1-\mu^\ast)m^\ast}=\hat q,
\end{align*}
namely
\begin{align*}
m^\ast=\frac{\mu^\ast(\mu^\ast(1-\tbias)+\tbias-\hat q)}{(1-\mu^\ast)\bigl(\hat q-(1-\alpha^\ast)\mu^\ast\bigr)}.
\end{align*}

Since $\nu(m;\mu^\ast)$ is decreasing in $m$, the scheme $\pi_m$ is persuasive for $\mathrm{High}$ iff $m\le m^\ast$. Therefore, although the primitive unknown is the pair $(\alpha^\ast,\mu^\ast)$, within the family $\{\pi_m\}_{m\in[0,1]}$ the sender's problem reduces to identifying the single boundary $m^\ast$. The two thresholds are linked by $\nu^\ast=\frac{\mu^\ast}{\mu^\ast+(1-\mu^\ast)m^\ast}$, or equivalently $m^\ast=\frac{\mu^\ast(1-\nu^\ast)}{(1-\mu^\ast)\nu^\ast}$. Thus $\nu^\ast$ is the critical threshold in posterior space, while $m^\ast$ is the corresponding threshold in the implementable signaling family.

\begin{algorithm}[h]
\caption{Safe Exploration for Unknown Bias and Unknown Prior}
\label{alg:setc_joint_unknown}
\begin{algorithmic}[1]
\State \textbf{Input:} horizon $T$, threshold $\hq$.
\State \textbf{Initialize:} interval $[a,b]\gets[0,1]$, step size $\varepsilon\gets \tfrac12$, time $t\gets 1$.

\State \textbf{Stage 1: Safe exploration}
\While{$b-a>T^{-1}$ \textbf{and} $t\le T$}
    \State $m \gets a$, \quad $m_{\mathrm{prev}}\gets a$
    \While{$m\le b$ \textbf{and} $t\le T$}
        \State Commit to the signaling scheme $\pi_m$ defined by
        \begin{align*}
        \pi_m(\mathrm{High}\mid \omega=1)=1,
        \qquad
        \pi_m(\mathrm{High}\mid \omega=0)=m
        \end{align*}
        \State Realize signal $s_t\in\{\mathrm{Low},\mathrm{High}\}$
        \If{$s_t=\mathrm{Low}$}
            \State $t\gets t+1$
            \State \textbf{continue}
        \Else
            \State Observe receiver action $a_t\in\{0,1\}$
            \If{$a_t=1$}
                \State $m_{\mathrm{prev}}\gets m$, \quad $m\gets m+\varepsilon$, \quad $t\gets t+1$
            \Else
                \State $[a,b]\gets[m_{\mathrm{prev}},\, m]$, \quad $\varepsilon\gets \varepsilon^2$, \quad $t\gets t+1$
                \State \textbf{break}
            \EndIf
        \EndIf
    \EndWhile
    \If{$m>b$}
        \State $[a,b]\gets[m_{\mathrm{prev}},\, b]$, \quad $\varepsilon\gets \varepsilon^2$
    \EndIf
\EndWhile

\State \textbf{Stage 2: Commit}
\State $\hat m \gets a$
\State Commit to the signaling scheme $\pi_{\hat m}$ for all remaining rounds
\end{algorithmic}
\end{algorithm}

\begin{proposition}[Binary-case regret upper bound for jointly unknown prior and bias]
\label{prop:joint-unknown-loglog}
Consider the binary model described above. Assume:
\begin{enumerate}
    \item $\Omega = A = \{0,1\}$;
    \item $u_S(a,\omega) = \mathbf{1}\{a=1\}$;
    \item the receiver takes action $1$ if and only if the distorted posterior is at least $\hq$, where $0<\mu^\ast<\hq<1$;
    \item persuasion is feasible, namely
    \begin{align*}
    \alpha^\ast \geq\alpha_{min}=\frac{\hq-\mu^\ast}{1-\mu^\ast}
    \end{align*}
    so that $m^\ast \in [0,1]$ is well-defined;
    \item ties are broken in favor of action $1$.
\end{enumerate}
Let $\Pi^{\mathrm{SEJ}}$ denote Algorithm~\ref{alg:setc_joint_unknown}. Then, for every $T\ge 4$,
\begin{align*}
\ireg{T}{}{\Pi^{\mathrm{SEJ}}}{\alpha^\ast,\mu^\ast}=O(\log\log T).
\end{align*}
In particular,
\begin{align*}
\ireg{T}{}{\Pi^{\mathrm{SEJ}}}{\alpha^\ast,\mu^\ast}
\le
\left(\frac{1-\mu^\ast}{\mu^\ast}+1\right)
\left(2+\left\lceil \log_2 \log_2 T \right\rceil\right)+1.
\end{align*}
\end{proposition}
\begin{proof}
Let $V^\ast$ denote the full-information benchmark value. As noted above, when $(\tbias,\mu^\ast)$ is known, an optimal scheme is supported on $\{0,\nu^\ast\}$, so $V^\ast=\frac{\mu^\ast}{\nu^\ast}=\mu^\ast+(1-\mu^\ast)m^\ast$.

Fix $m\in[0,1]$ and consider $\pi_m$. The probability of the $\mathrm{High}$ signal is $p_H(m)=\mu^\ast+(1-\mu^\ast)m$. By definition of $m^\ast$, the $\mathrm{High}$ recommendation is persuasive if and only if $m\le m^\ast$. Hence the sender's one-round utility is $V(m)=\mu^\ast+(1-\mu^\ast)m$ if $m\le m^\ast$, and $V(m)=0$ otherwise. Therefore the one-round regret is $r(m)=(1-\mu^\ast)(m^\ast-m)$ if $m\le m^\ast$, and $r(m)=V^\ast$ if $m>m^\ast$.

Viewing one test point $m$ as an epoch that lasts until the first $\mathrm{High}$ realization, the epoch length is geometric with mean $\mathbb{E}[N(m)]=1/p_H(m)$. Thus the expected regret of that epoch is $R_{\mathrm{epoch}}(m)=r(m)/p_H(m)$. If $m\le m^\ast$, then
\begin{align*}
R_{\mathrm{epoch}}(m)
=\frac{(1-\mu^\ast)(m^\ast-m)}{\mu^\ast+(1-\mu^\ast)m}
\le \frac{1-\mu^\ast}{\mu^\ast}(m^\ast-m).
\end{align*}
If $m>m^\ast$, then $p_H(m)\ge p_H(m^\ast)=V^\ast$, so
\begin{align*}
R_{\mathrm{epoch}}(m)=\frac{V^\ast}{p_H(m)}\le 1.
\end{align*}

Now consider the exploration phase. The maintained interval $[a,b]$ always satisfies $a\le m^\ast\le b$: this is true initially since $[a,b]=[0,1]$, and each update preserves it because acceptance is equivalent to $m\le m^\ast$. Let $L_k$ be the interval length at the start of phase $k$. Since $L_0=1$ and $\varepsilon_0=1/2$, after phase $0$ we have $L_1=1/2$. For every $k\ge 1$, the new interval length equals the previous step size and the algorithm updates $\varepsilon\gets\varepsilon^2$, so $\varepsilon_k=L_k^2$ and $L_{k+1}=L_k^2$. Hence $L_k=(1/2)^{2^{k-1}}$ for $k\ge 1$. Exploration stops once $L_k\le T^{-1}$, so the number of phases is at most
\begin{align*}
P\le 2+\left\lceil \log_2\log_2 T\right\rceil.
\end{align*}

Fix one phase with interval $[a,b]$, length $L=b-a$, and step size $\varepsilon$. Let the accepted test points be $m_j=a+j\varepsilon$ for $j=0,\dots,J-1$, where $J\le L/\varepsilon$, and possibly there is one rejected point $m_J$. Then
\begin{align*}
\sum_{j=0}^{J-1} R_{\mathrm{epoch}}(m_j)
\le \frac{1-\mu^\ast}{\mu^\ast}\sum_{j=0}^{J-1}(m^\ast-m_j).
\end{align*}
Since rejection first occurs at $m_J$, we have $m^\ast<a+J\varepsilon$, so $m^\ast-m_j\le (J-j)\varepsilon$. Therefore
\begin{align*}
\sum_{j=0}^{J-1}(m^\ast-m_j)
\le \varepsilon\sum_{r=1}^J r
= \varepsilon\frac{J(J+1)}{2}
\le \frac{L^2}{2\varepsilon}+\frac{L}{2}.
\end{align*}
Hence the regret from accepted points in this phase is at most $\frac{1-\mu^\ast}{\mu^\ast}\left(\frac{L^2}{2\varepsilon}+\frac{L}{2}\right)$. For phase $0$, this is at most $\frac{3}{2}\cdot \frac{1-\mu^\ast}{\mu^\ast}$. For every later phase, $\varepsilon=L^2$, so it is at most $\frac{1-\mu^\ast}{\mu^\ast}\left(\frac12+\frac{L}{2}\right)\le \frac{1-\mu^\ast}{\mu^\ast}$. Each phase has at most one rejected point, and its epoch regret is at most $1$. Therefore
\begin{align*}
R_{\mathrm{explore}}
&\le \frac{3}{2}\cdot \frac{1-\mu^\ast}{\mu^\ast}+1+(P-1)\left(\frac{1-\mu^\ast}{\mu^\ast}+1\right) \\
&\le \left(\frac{1-\mu^\ast}{\mu^\ast}+1\right)P.
\end{align*}

In the commitment phase, the algorithm commits to $\hat m=a$. By the interval invariant, $\hat m\le m^\ast$ and $m^\ast-\hat m\le b-a\le 1/T$. Thus $\pi_{\hat m}$ is persuasive, and its per-round regret is $(1-\mu^\ast)(m^\ast-\hat m)\le (1-\mu^\ast)/T\le 1/T$. So the total commitment regret is at most $1$.

Combining the two parts,
\begin{align*}
\ireg{T}{}{\Pi^{\mathrm{SEJ}}}{\alpha^\ast,\mu^\ast}
&\le \left(\frac{1-\mu^\ast}{\mu^\ast}+1\right)P+1 \\
&\le \left(\frac{1-\mu^\ast}{\mu^\ast}+1\right)\left(2+\left\lceil \log_2\log_2 T\right\rceil\right)+1.
\end{align*}
This proves the proposition.
\end{proof}

\subsection{General case with jointly unknown prior and bias}
\begin{proposition}[General-case regret lower bound for jointly unknown prior and bias]\label{prop:joint_general_lower}
In the jointly unknown prior and bias model with finite state and action spaces, for any learning algorithm $\Pi$, there exists an instance such that
\begin{align*}
\ireg{T}{}{\Pi}{\alpha^\ast,\mu^\ast}=\Omega(\log T).
\end{align*}
\end{proposition}

\begin{proof}
The jointly unknown model strictly contains the unknown-prior model studied in \cite{Li_Lin_2025} as the special case $\alpha^\ast=1$. Indeed, when $\alpha^\ast=1$, the distorted posterior reduces to the Bayesian posterior.

Therefore any learning algorithm for the jointly unknown model induces a learning algorithm for the unknown-prior model. By the $\Omega(\log T)$ lower bound for the unknown-prior problem in \cite{Li_Lin_2025}, the same lower bound must hold for the jointly unknown model.
\end{proof}

\paragraph{Why the general jointly unknown case is difficult.}
Although the lower bound $\Omega(\log T)$ follows immediately from the unknown-prior model by restricting to the special case $\alpha^\ast=1$, obtaining a matching $O(\log T)$ upper bound for the general jointly unknown case appears substantially harder. The main difficulties are as follows.
\begin{itemize}
    \item \textbf{The identification step from unknown-prior model no longer decouples into pairwise prior-ratio estimation.}
    
    In the unknown-prior model of \cite{Li_Lin_2025}, the sender learns $\mu^\ast$ by constructing experiments in which the receiver's action reveals the sign of a linear expression involving only a pair of states. This makes it possible to estimate ratios such as $\mu^\ast(\omega_i)/\mu^\ast(\omega_j)$ and then reconstruct the prior. In the jointly unknown model, however, after a signal inducing Bayesian posterior $\nu$, the receiver evaluates actions using the distorted belief
    $
    (1-\alpha^\ast)\mu^\ast + \alpha^\ast \nu$.
    Hence, for any two actions $a,a'$, the receiver compares
    \begin{align*}
    \sum_{\omega\in\Omega}
    \Bigl((1-\alpha^\ast)\mu^\ast(\omega)+\alpha^\ast \nu(\omega)\Bigr)
    \bigl(u_R(a,\omega)-u_R(a',\omega)\bigr).
    \end{align*}
    Since $\nu$ itself depends on $\mu^\ast$ through Bayes' rule, the sign of this expression depends jointly on the entire prior vector $\mu^\ast$ and on $\alpha^\ast$. Therefore a single experiment no longer isolates a pairwise ratio of prior masses. In particular, the ratio-estimation argument from \cite{Li_Lin_2025} cannot be applied directly. 

    

    \item \textbf{The binary-action reduction from \cite{Li_Lin_2025} does not extend directly once bias is also unknown.}
    
    In the binary-action setting of \cite{Li_Lin_2025}, the sender's optimization problem can be written as a fractional knapsack problem whose value-to-weight ordering is independent of the prior, because the common factor $\mu^\ast(\omega)$ cancels in the ratio. This allows \cite{Li_Lin_2025} to define a universal one-parameter family $\{\pi_M\}$ and to characterize persuasion by a single threshold $M^\ast$. In the jointly unknown biased model, however, the persuasive constraint for recommending action $1$ takes the form
    \begin{align*}
    \sum_{\omega\in\Omega}\mu^\ast(\omega)x_\omega
    \Bigl(\alpha^\ast g_\omega + (1-\alpha^\ast)B(\mu^\ast)\Bigr)\le 0,
    \end{align*}
    where $x_\omega=\pi(1\mid \omega)$, $g_\omega=u_R(0,\omega)-u_R(1,\omega)$, and $B(\mu^\ast)=\sum_{\omega\in\Omega}\mu^\ast(\omega)g_\omega$.
    Thus the effective ``weight'' of state $\omega$ is
    \begin{align*}
    w_\omega(\alpha^\ast,\mu^\ast)
    =
    \alpha^\ast g_\omega + (1-\alpha^\ast)B(\mu^\ast),
    \end{align*}
    which depends jointly on $\alpha^\ast$ and on the full prior vector $\mu^\ast$. Consequently, the state ordering relevant for the greedy argument is no longer fixed ex ante, and there is no direct analogue of the universal family $\{\pi_M\}$ used in \cite{Li_Lin_2025}. Therefore the $O(\log\log T)$ binary-action argument of \cite{Li_Lin_2025} does not extend mechanically to the jointly unknown setting. 

    \item \textbf{The robustification step from the unknown-prior paper is not directly available under joint misspecification of prior and bias.}
    
    The $O(\log T)$ algorithm in \cite{Li_Lin_2025} has an exploitation phase based on robustification: once the sender obtains an estimate $\hat\mu$ of the true prior $\mu^\ast$, she computes a signaling scheme that is optimal for $\hat\mu$ and then modifies it so that it remains persuasive and near-optimal for all priors in a neighborhood of $\hat\mu$. This argument is tailored to perturbations in the prior alone. In the jointly unknown model, however, persuasion is determined by distorted beliefs of the form $
    (1-\alpha)\mu + \alpha \nu$.
    
    Therefore, if the sender only knows an approximate pair $(\hat\mu,\hat\alpha)$, then to carry out an exploitation phase one would need a \emph{joint robustification} result: namely, a theorem showing that a scheme designed for $(\hat\mu,\hat\alpha)$ can be modified so as to remain persuasive and near-optimal uniformly over all nearby pairs $(\mu,\alpha)$. 
\end{itemize}
For these reasons, the techniques of the two papers do not currently yield a general $O(\log T)$ algorithm.
    The lower bound $\Omega(\log T)$ is inherited immediately from the unknown-prior model, but the corresponding upper-bound techniques rely on tools that are specific either to prior uncertainty alone or to a one-dimensional threshold structure. In the general jointly unknown case, neither of these simplifications is currently available. Therefore, while an $O(\log T)$ algorithm may still exist, establishing such a result would require new ideas beyond the current identification and robustification arguments. For this reason, we view the general jointly unknown prior and bias problem as an open question.
\newpage

\end{document}